\documentclass[12pt,a4paper]{article}
\usepackage[utf8]{inputenc}
\usepackage{amsmath}
\usepackage{caption}

\usepackage{hyperref}
\usepackage{graphicx}
\usepackage[margin=1in]{geometry}
\usepackage{cite}
\usepackage{hyperref}
\usepackage{graphicx} 		
\usepackage{eurosym}
\usepackage{hyperref}
\usepackage{enumitem}
\usepackage{amssymb}

\usepackage{nopageno}
\usepackage{enumitem}
\usepackage{url}

\usepackage{amsmath}
\usepackage{amsfonts}
\usepackage{amssymb}
\usepackage{graphicx}
\usepackage{longtable}
\usepackage{subfigure}

\usepackage[affil-it]{authblk}

\usepackage{booktabs}
\date{}
\usepackage[toc,page]{appendix}
\usepackage{authblk}
\usepackage{amsmath}
\usepackage{amssymb}
\usepackage{amsthm}
\usepackage[pro]{fontawesome5}
\usepackage{hyperref}

\newtheorem{theorem}{Theorem}
\newtheorem{lemma}{Lemma}
\newcommand{\norm}[1]{\left\lVert#1\right\rVert}

\newenvironment{proof}{\paragraph{Proof:}}{\hfill$\square$}
\begin{document}

\title{Kernel Biclustering algorithm in Hilbert Spaces}

\author[1]{Marcos Matabuena}
\author[1]{J.C Vidal}
\author[2]{Oscar Hernan Madrid Padilla}
\author[3]{Dino Sejdinovic}

\affil[1]{Centro Singular de Investigación en Tecnoloxías Intelixentes (CiTIUS), Universidade de Santiago de Compostela, 15782 Santiago de Compsotela, Spain.}
\affil[2]{ Department of Statistics at University of California, Los Angeles}
\affil[2]{ Department of Statistics
University of Oxford}

\date{Dated: \today}

\maketitle

\begin{abstract}

	Biclustering algorithms partition data and covariates simultaneously, providing new insights in several domains, such as analyzing gene expression to discover new biological functions. This paper develops a new model-free biclustering algorithm in abstract spaces using the notions of energy distance (ED) and the maximum mean discrepancy (MMD) -- two distances between probability distributions capable of handling complex data such as curves or graphs. The proposed method can learn more general and complex cluster shapes than most existing literature approaches, which usually focus on detecting mean and variance differences. Although the biclustering configurations of our approach are constrained to create disjoint structures at the datum and covariate levels, the results are competitive. Our results are similar to state-of-the-art methods in their optimal scenarios, assuming a proper kernel choice, outperforming them when cluster differences are concentrated in higher-order moments. The model's performance has been tested in several situations that involve simulated and real-world datasets. Finally, new theoretical consistency results are established using some tools of the theory of optimal transport. 
	
\end{abstract}

\section{Introduction}

Cluster analysis plays a significant role in the exploratory and descriptive analysis of data in various unsupervised learning tasks and its application to multiple real-world problems \cite{jain1999data, ahlqvist2018novel}. 
However, the use of traditional cluster algorithms  can be limited in many real-world applications. For instance, data distribution may be specified by latent local structures, in which the units or individuals of each cluster are constituted using a different set of active covariates. For example, in the analysis of genetic profiles in bioinformatics,  patient sub-types are characterized by specific regions of genes that can even show disjoint gene expression patterns among clusters \cite{madeira2004biclustering}.

Biclustering algorithms \cite{madeira2004biclustering,busygin2008biclustering} are promising to overcome the aforementioned limitation as they can create heterogeneous data groups with different covariates between clusters of individuals. In this way, they increase the interpretability and clinical meaning of the groups, making the analysis more robust to noise and mitigating the curse of dimensionality.

Biclustering methods can be classified mainly into two categories according to the structures they are able to build \cite{madeira2004biclustering,fraiman2020biclustering, laclau2017co}. In the first category, clusters are arbitrarily positioned, and can overlap with each other \cite{cheng2000biclustering,getz2000coupled,bergmann2003iterative,tanay2004revealing,lazzeroni2002plaid,ben2002discovering,shabalin2009finding}, while, in the second one, overlapping is not allowed and, thus, clusters are formed following a checkerboard structure in their matrix representation \cite{kluger2003spectral,lee2010biclustering,chi2017convex,flynn2020profile,chen2013biclustering,cho2004minimum,dhillon2001co}.

According to their probabilistic nature, one can consider biclusterning based on generative hierarchical models specified through parametric distributions \cite{flynn2020profile, gu2008bayesian,zhang2010bayesian}. On the other hand, we can consider non-parametric biclustering algorithms defined through ANOVA decomposition, where clusters are constructed using different distance criteria that measure the variability of the clusters (cf. e.g. \cite{bryan2006application}).

In addition, sparse biclustering algorithms have recently received more attention \cite{helgeson2020biclustering}. Sparse biclustering algorithms usually perform well in high-dimensional sparse structures, as they simultaneously perform  variable selection to reduce the dimensionality of the problem and the construction of the cluster.

However, despite progress in this active research area, existing methods still have limitations. For example, like in the $k$-means or probabilistic mixture algorithms, many biclustering algorithms can only detect structural changes in mean or the constructed cluster shapes following particular geometries according to parametric distributions \cite{flynn2020profile}. Nevertheless, variance \cite{chen2013biclustering} or another higher-order moment may be a critical component in real-world problems such as in establishing biological differences across patient subtypes \cite{de2019gene,helgeson2020biclustering}. Furthermore, it is well-known that solving a biclustering problem is generally an NP-hard problem \cite{madeira2004biclustering} and it can be challenging to obtain efficient and reliable optimization strategies on a large scale or high dimensional settings. Finally, existing biclustering techniques with complex statistical objects are limited to recent Euclidean functional data contributions \cite{galvani2021funcc}. However, methods for complex statistical objects can greatly interest contemporary applications such as personalized or precision medicine. In these contexts, the increasing ability to register patient health at high-order resolution calls for better representations to have a more reliable characterization of clinical patient conditions \cite{matabuena2021glucodensities, matabuena2021distributional}.

Kernel algorithms based on reproducing kernel Hilbert spaces (RKHS) learning paradigm are a powerful data analysis modeling strategy, both in unsupervised as well as in supervised learning, to integrate and analyze complex statistical objects such as dynamic structures or functional data objects \cite{dubey2020functional}. A representative example of the modeling power of these techniques is the maximum mean discrepancy (MMD) \cite{gretton2012kernel}, a distance between probability measures that has been successfully applied to many problems \cite{romano2020deep,gretton2012kernel,wu2020minimax}. Recently in \cite{francca2020kernel}, MMD  has been used to define a new non-parametric clustering algorithm in separable Hilbert spaces called kernel $k$-groups, which improves the performance in several situations examined over traditional methods. 

In this paper, we extend the kernel $k$-groups approach of \cite{francca2020kernel} to biclustering, defining the first general biclustering method ($AKKB$) in separable Hilbert spaces based on the RKHS paradigm. Our method generalizes other clustering algorithms, such as the $k$-means or graph partitioning-based algorithms. Unlike most of the  existing methodologies, construction of clusters that take into account features beyond  mean and variance (or other simple parametrizations).

\subsection{Outline of contributions}

We summarize the main methodological contributions of this paper below:

\begin{enumerate}	
	\item To the best of our knowledge, we propose the first biclustering methodology under the RKHS learning paradigm that can be used to analyze complex data such as functional curves or graph structures. In our applications, we illustrate this properly using functional data extracted from continuous glucose technology in a sample of diabetes patients.   At the same time, the proposed methodology allows for detecting distributional differences beyond low order moments, increasing the variety of cluster shapes that algorithms can discover. 
	
	\item We propose an efficient optimization strategy that iteratively applies the kernel $k$-groups algorithm on the individual and covariate levels. For this purpose, we introduce the structural assumption that the clusters are mutually disjoint at both levels. Unlike most existing biclustering literature based on greedy optimization strategies (see for example \cite{bryan2006application}), our new optimization strategy comes equipped with theoretical guarantees. 

	\item  We establish statistical consistency of our algorithm, and we discuss  mathematical connections of our proposal with other existing clustering algorithms. 
	\end{enumerate}

\subsection{Notation}\label{sec:previa}

Let $(\Omega, \mathcal{F}_{i}, \mathbb{P}_{i})$,  $i\in \mathcal{I}=\left \{1,\dots, p\right\}$, be $p$ probability spaces with common sample space $\Omega$. For each $i\in \mathcal{I}$,  let $\mathbb{H}_i$ be a separable Hilbert space over $\mathbb{R}$ with inner product $\left \langle\cdot, \cdot \right\rangle_{\mathbb{H}_i}$, norm $\norm{\cdot}_{\mathbb{H}_i}$, and $\sigma$-field of sets $\mathcal{B}_{\mathbb{H}_i}$.  Let $X^i: \Omega \to \mathbb{H}_i$ be a random element in $\mathbb{H}_i$, that is, an  $(\mathcal{F}_{i},\mathcal{B}_{\mathbb{H}_i})$-measurable mapping. Along this paper, we require that $\mathbb{E} (\norm{X^i}^{2}_{\mathbb{H}_{i}}) < \infty$ and, for simplicity we assume that Hilbert spaces $\mathbb{H}_i$'s are identical, and denoted as $\mathbb{H}$. Thus, multivariate random variables are denoted as $X = (X^{1}, \ldots, X^{p})\in \mathbb{H}^{p}= \mathbb{H} \times \cdots \times  \mathbb{H}$ and the random sample observed as $X^{data}=\{X_i\}_{i=1}^{n}$, in which each $X_i$ $(i=1,\dots,n)$ is an independent copy of random variable $X$.

Let us denote by $\mathcal{J}= \{1,\ldots, n \}$ the set of indices of the $n$ observations and by $\mathcal{I}= \{1,\ldots, p\}$ the set of indices of the $p$ covariates. Since our goal is to simultaneously partition the set $\mathcal{J}$ and $\mathcal{I}$ into a fixed number of $m$ disjoints clusters, we denote by $\mathcal{A}:=\{\mathcal{I}_{1},\dots, \mathcal{I}_{k}\}$ and $\mathcal{B}: =\{\mathcal{J}_{1},\dots,\mathcal{J}_{k}\}$ the arbitrary partitions of the previous sets, i.e., $ \cup^{k}_{i=1} \mathcal{I}_i= \mathcal{I}$,  $\cup^{k}_{i=1} \mathcal{J}_{i}= \mathcal{J}$ and, $\forall i\neq j$ $\mathcal{I}_{i} \cap \mathcal{I}_{j} = \emptyset$, $\mathcal{J}_{i} \cap \mathcal{J}_{j} = \emptyset$.

For a given $x\in  \mathbb{H}^{p}$, and the covariates selected in $\mathcal{I}_j$, let $x\left(\mathcal{I}_j\right)= \left(x^{l}\right)_{l\in \mathcal{I}_j}\in \mathbb{H}_{\mathcal{I}_j}$, where $\mathbb{H}_{\mathcal{I}_j}$ is a functional space in the Cartesian product of topological spaces $\times_{i\in \mathcal{I}_j} \mathbb{H}$.  With this notation in hand, we define the norm of $\norm{x\left(\mathcal{I}_j\right)}_{\mathcal{I}_j}$ as $\norm{x\left(\mathcal{I}_j\right)}_{\mathcal{I}_j}=\sqrt{\frac{\sum_{i\in \mathcal{I}_j } \langle x^i, x^i \rangle_{\mathbb{H}}}{\vert \mathcal{I}_{j} \vert}}= \sqrt{\frac{\sum_{i\in \mathcal{I}_j } \norm{x^i}^{2}_{\mathbb{H}}}{\vert \mathcal{I}_{j} \vert}}$. For example, suppose that
$$\mathbb{H}= L^{2}\left(\left[0,1\right]\right)= \{f:\left[0,1\right]\to\mathbb{R}:f \text{ is measurable and} \int_{0}^{1} f^{2}\left(t\right)dt < \infty \}.$$
Then, $\norm{x\left(\mathcal{I}_j\right)}_{\mathcal{I}_j}= \sqrt{\frac{\sum_{i\in \mathcal{I}_j } \int_{0}^{1}\left(x^{i}\left(t\right)\right)^{2} dt}{\vert \mathcal{I}_{j} \vert}} $. We note that the defined norm already includes appropriate normalization by the number of components in each vector. This is for notational convenience as we calculate the sum of global distances as sum of the individual contribution of the distances of each marginal space.

The biclustering algorithm introduced here assumes that partitions at individual and covariate levels are mutually disjoint. Then, in order to  express in a compact way the underlying optimization problem, we define the set $\mathcal{C}^{\mathcal{A},\mathcal{B}}:=\{\mathcal{C}_1,\dots, \mathcal{C}_k\}$, where each $\mathcal{C}_{l}= \{ X_i(\mathcal{I}_l): i \in \mathcal{J}_{l}\}$ depends on the selected partitions $\mathcal{A}$ and $\mathcal{B}$. Our biclustering algorithm measures the similarity between  individuals
within a subset of covariates using a symmetric positive definite kernel, $k^{\mathcal{I}_j}: \mathbb{H}_{\mathcal{I}_j}  \times \mathbb{H}_{\mathcal{I}_j} \to \mathbb{R}^{+}$. We assume that  $k^{\mathcal{I}_j}$'s can be expressed as  $k^{\mathcal{I}_j}(x(\mathcal{I}_j),y(\mathcal{I}_j))= f(\norm{x(\mathcal{I}_j)-y(\mathcal{I}_j)}_{\mathcal{I}_j})$, where $f(\cdot)$ is a monotonic continuous function. Along this paper, we will use the Gaussian kernel, $k^{\mathcal{I}_j}(x\left(\mathcal{I}_j\right),y\left(\mathcal{I}_j\right))= e^{-\frac{\norm{x (\mathcal{I}_j)-y(\mathcal{I}_j)}^{2}_{\mathcal{I}_j}}{\sigma^{2}}}$, where $\sigma>0$ is the bandwidth parameter of the kernel. We denote the RKHS related to $k^{\mathcal{I}_j}$ as $\mathcal{H}^{\mathcal{I}_j}$ and their norm $\norm{\cdot}_{\mathcal{H}^{\mathcal{I}_j}}$.

We note that each of the local kernels $k^{\mathcal{I}_j}$ used are universal and characteristic  \cite{sriperumbudur2011universality, christmann2010}. Consequently,  we can characterize the equality in the distribution or statistical independence between random samples (cf. \cite{sriperumbudur2011universality} for more details). 


\subsection{Outline of the paper}

The structure of the paper is as follows. First, Section \ref{sec:distancias} introduces the mathematical foundations of statistical distances between random variables in RKHS that constitute the formal tools to design our biclustering algorithm. Bellow, Section \ref{sec:clusteringnormal} discuses  existing literature on   clustering  algorithms on RKHS that are based on our biclustering algorithm. Subsequently, Section \ref{sec:biclustering} introduces the formal definition of our biclustering algorithm together with some numerical schemes to solve the underlying optimization problem. Next, section \ref{sec:theory} introduces the theoretical statistical results of our new biclustering algorithm. Then, Section \ref{sec:results} introduces the results with simulated and real databases, in which we show the promising results of our proposal. Finally, in Section \ref{sec:discusion} we discuss our main findings, and potential future research directions.   

\section{Mathematical models}

\subsection{Preliminary models}

\subsubsection{Maximum mean discrepancy and energy distance}\label{sec:distancias}

Maximum mean discrepancy (MMD) and energy distance (ED) are two equivalent families of statistical distances between probability measures supported on separable Hilbert Spaces. With their increase in popularity at the beginning of this century, data analysis methods derived from these distances are widely adopted in many applications and statistical modeling tasks, e.g., hypothesis testing, cluster analysis, or change-point detection problems. Given the mentioned equivalence,  along with this paper, we restrict our attention to the notion of maximum mean discrepancy that is represented in terms of a symmetric and positive definite kernel function $k: \mathbb{H}^{p}\times \mathbb{H}^{p}\to \mathbb{R}^{+}$  corresponding to an RKHS $\mathcal{H}$.


Let $X\sim F$ and $Y\sim G$ be two $\mathbb{H}^{p}$-random variables that satisfy  $\mathbb{E}\left(\norm{X}^{2}\right)<\infty$ and  $\mathbb{E}\left(\norm{Y}^{2}\right)<\infty$. The MMD can be defined as:

\begin{align}\label{eqn:independencia2}
	MMD(X,Y)^{2}=  \norm{\phi_{X}-\phi_{Y}}^{2}_{\mathcal{H}} = \norm{\int k(x,\cdot)F(dx)-\int k(y,\cdot)G(dy)}^{2}_{\mathcal{H}}= \\  = \mathbb{E}(k(X,X^{\prime}))+\mathbb{E}(k(Y,Y^{\prime}))-2\mathbb{E}(k(X,Y)),
\end{align}

where  $\phi_{X}\in \mathcal{H}$ is known in the literature as the kernel mean embedding of the random variable $X$. For any distribution function $F$, if the mapping $\phi_{X}:x\in \mathbb{H}^{p} \to \int k\left(x,\cdot\right)F\left(dx\right)\in \mathcal{H}$ is injective, then we can construct omnibus tests that characterize the equality in distribution \cite{gretton2012kernel}. 

Given two samples i.i.d $\left\{X_i\right\}_{i=1}^{n_1}\sim F$ and $\left\{Y_i\right\}_{i=1}^{n_2}\sim G$, $n_1+n_2=n$, we can straightforwardly estimate the empirical counterparts, that we denote as  $\hat{F}$ and $\hat{G}$, respectively. Using this estimate, we estimate the statistics Equation $1$, using the following estimator

\begin{equation}
	\widetilde{MMD}\left(X,Y\right) =   \frac{1}{n_1^2} \sum_{i=1}^{n_1} \sum_{j=1}^{n_1}  k\left(X_i,X_j\right)  + \frac{1}{n_2^2} \sum_{i=1}^{n_2} \sum_{j=1}^{n_2}  k\left(Y_i,Y_j\right) -  \frac{2}{n_1 n_2} \sum_{i=1}^{n_1} \sum_{j=1}^{n_2} k\left(X_i,Y_j\right).
\end{equation}


The choice of the kernel function $k$  in the  ED or MMD is critical in the finite sample performance in the different modeling tasks. However, it is not easy to establish a criterion in,  general, since each selected semi-metric characterizes distributional differences, giving more or less priority to a specific moment in the distance computation. This information may not be available in practice, and may be difficult to obtain from expert knowledge or prior studies. \cite{GreSriSejStrBalPonFuk2012, jitkrittum2016} establish principles for kernel choice on the basis of maximizing the two-sample test power.

The notion of empirical ED and MMD estimators can be extended to the multi-sample setting using a global statistic that considers the distances between the pairs of random samples (see for example \cite{rizzo2010disco}). Consider $m$ $(m>2)$ random variables, $X^{1}\sim F^{1},\dots,X^{m}\sim F^{m}$ and $m$ i.i.d random samples,  
$X^{1,data}=\{X^{1,data}_i\}^{n_1}_{i=1}\sim F^{1}$, $\dots$, $X^{m,data}=\{X^{m,data}_i\}^{n_m}_{i=1}\sim F^m$, $n_1+\dots+n_m=n$. We denote the multi-sample  ED as \\ 

\begin{align}
\widetilde{\epsilon}_{\rho}\left(X^1,X^2,\dots, X^m\right)= \sum_{a=1}^{m} \sum_{b \neq a}^{m} \frac{2 n_a n_b}{n} 
  \Bigg[
	\frac{2}{n_a n_b} \sum_{i=1}^{n_a} \sum_{j=1}^{n_b}  \rho\left(X^{a}_i,X^{b}_j\right) \\ - 
	\frac{1}{n_m^2} \sum_{i=1}^{n_a} \sum_{j=1}^{n_b}  \rho\left(X^{a}_i,X^{b}_j\right)  - \frac{1}{n_r^2} \sum_{i=1}^{n_a} \sum_{j=1}^{n_b}  \rho\left(X^{a}_i,X^{b}_j\right)\Bigg], 
\end{align}

where $\rho: \mathbb{H}^{p} \times \mathbb{H}^{p} \to \mathbb{R}$ is a semi-metric of negative-type that is connected with the kernel $k\left(\cdot, \cdot\right),$ $k\left(x,y\right)= \frac{1}{2}\left[ \rho\left(x,x_0\right)+\rho\left(y, x_0\right)-\rho\left(x,y\right)\right]$, where  $x_0\in \mathcal{H}^{p}$ is a fixed-point. For example, if $\rho\left(x,y\right)= \norm{x-y}$, $x_0=0\in \mathbb{H}^{p}$, then $k\left(x,y\right)= \frac{1}{2}\left[\norm{x}+\norm{y}-\norm{x-y}\right].$

\subsubsection{Model-free clustering 
in RKHS}\label{sec:clusteringnormal}

The purpose of any  clustering algorithm is to create $m$ groups such that elements in the same group are similar, while elements between different groups are different from each other. Following \cite{francca2020kernel}, a natural clustering algorithm can be defined with ED by selecting the partition of $X^{data}$, the size $m$, and the set of clusters $\{\hat{\mathcal{C}}_1,\dots, \hat{\mathcal{C}}_m\}$ that maximize the following optimization problem with multi-sample energy statistics:

\begin{equation}
\label{eqn:energia3b}
  \{\hat{\mathcal{C}}_1,\dots, \hat{\mathcal{C}}_m\}=  \arg \max_{\{\mathcal{C}_1,\dots, \mathcal{C}_m\}} \widetilde{\epsilon}_{\rho}\left(\mathcal{C}_1,\dots, \mathcal{C}_m\right).
\end{equation}
  
According to a semi-metric $\rho$ selected beforehand, this optimization problem finds the partition of $X^{data}$ which maximizes the group-separation criteria based on multi-sample energy distance statistics.

It is known that given a semi-metric $\rho$ of negative type, one can associate a symmetric, positive definite,  kernel $k$, as follows, $k\left(x,y\right)= \frac{1}{2}\left[ \rho\left(x,x_0\right)+\rho\left(y, x_0\right)-\rho\left(x,y\right)\right]$, $\forall x,y \in \mathbb{H}^{p}$, where  $x_0 \in \mathbb{H}^{p}$ is an arbitrary fixed point \cite{berg1984harmonic,sejdinovic2013equivalence}. Following \cite{francca2020kernel}, and using the above relation, the previous optimization problem can be written as:
\begin{equation}\label{eqn:energia4}
  \left\{\hat{\mathcal{C}}_1,\dots, \hat{\mathcal{C}}_m\right\}  = \arg \min_{\{\mathcal{C}_1,\dots, \mathcal{C}_m\}} \sum_{i=1}^{m} \frac{1}{n_{C_i}} \underbrace{ \sum_{x\in \mathcal{C}_{i}} 
  		\sum_{y\in \mathcal{C}_{i}} k\left(x,y\right)}_{Q_i},
\end{equation}
where $|\mathcal{C}_i|= n_{\mathcal{C}_i}$; and $Q_i$ is the so-called \textit{within Kernel dispersion} between the elements of the cluster $\mathcal{C}_i$. 
  
Using the optimization problem defined in Eq. (\ref{eqn:energia4}), we can see that classical clustering algorithms such as graph partitioning or kernel $k$-means are a special case of the optimization problem defined in Eq.  (\ref{eqn:energia3b}). In this paper, we are particularly interested in its relationship with the kernel $k$-means algorithm since it motivates our  biclustering  formulation and our optimization strategy. In particular, the optimization problem of Eq. (\ref{eqn:energia4}) is equivalent to:
\begin{equation}\label{formula:energia6}
   \left \{\hat{\mathcal{C}}_1,\dots, \hat{\mathcal{C}}_m \right\} =  \arg \min_{\left \{\mathcal{C}_1,\dots, \mathcal{C}_m\right \}} \sum_{i=1}^{m} \frac{1}{n_{\mathcal{C}_i}} \sum_{x\in \mathcal{C}_i}  \norm{\phi\left(x\right)-\hat{\mu}_{\mathcal{C}_i}}^{2}_{\mathcal{H}}, 
\end{equation}
where $\phi\left(x\right)= k\left(\cdot,x\right)\in \mathcal{H}$,  $\hat{\mu}_{\mathcal{C}_i}= \frac{1}{n_{\mathcal{C}_i}} \sum_{x\in \mathcal{C}_{i}}\phi\left(x\right)$. 
  
With this relationship in mind, we can find the optimal solution of kernel $k$-groups clustering algorithm \cite{francca2020kernel}, using the classical resolution strategies of kernel $k$-means algorithms such as Hartigan's or Lloyd's algorithms. The underlying idea is to define a mathematical criteria to dynamically assign the points to new clusters such that the value of the objective function decrease.   

Let $Q_{l}(x)= \sum_{y \in \mathcal{C}_{l}} k(x,y)$ be the cost to move the datum $x\in \mathbb{H}^{p}$ to cluster $\mathcal{C}_l$. We can define the cost to move the datum $x$ from cluster $j$ to $l$ as
\begin{align}
  \Delta Q^{(j\to l)}(x)= \frac{Q^{+}_l}{n_{\mathcal{C}_l}+1}+ \frac{Q^{-}_j}{n_{\mathcal{C}_j}-1}-\frac{Q_l}{n_{\mathcal{C}_l}}-\frac{Q_j}{n_{\mathcal{C}_j}},
\end{align}
where $Q^{+}_l(x)$ and $Q^{-}_j(x)$, are the cost to add $x$ to the cluster $\mathcal{C}_l$ and to remove $x$ from the cluster $\mathcal{C}_j$, respectively, where
\begin{align*}
  Q^{+}_l(x) &= Q_l+2Q_l(x)+k(x,x),\\
  Q^{-}_j(x) &= Q_j-2Q_j(x)+k(x,x),
\end{align*}
and
\begin{equation*}
    Q_l= \sum_{x\in \mathcal{C}_{l}} 
    \sum_{y\in \mathcal{C}_{l}} k\left(x,y\right) \hspace{0.2cm} (l=1,\dots,m),
\end{equation*}
is the so-called \textit{within Kernel dispersion} between the elements of the cluster $\mathcal{C}_l$.
A resolution strategy can therefore consist of solving recurrently the following problem:

\begin{align}
  j^{*} &= \arg \max_{l=1,\dots, m|l\neq j} \Delta Q^{ \left(j\to l\right)}\left(x\right), \forall x\in X^{data}.
\end{align}
If $\Delta Q^{\left(j\to j^{*}\right)}\left(x\right)>0$, that is, if  we can improve the objective function moving the datum $x$ to cluster $\mathcal{C}_{j^{*}}$.

\subsection{$AKKB$ algorithm}\label{sec:biclustering}

\subsubsection{Mathemathical definition}

This section introduces the new biclustering kernel algorithm in separable Hilbert spaces inspired by the concepts of ED and MMD. Its final formulation can be seen as an infinite-dimensional RKHS generalization of the biclustering algorithm proposed in \cite{fraiman2020biclustering}, that is designed to detect structural differences in mean in the context of finite-dimensional Euclidean spaces.

First, we note that the proposed algorithm introduces specific shapes constraints to construct $k$ clusters at individual and covariate levels. In particular, we assume that $k$ clusters formed in the previous two levels are mutually disjoint. Under this structural assumption, we provide an efficient resolution strategy using the standard kernel $k$-means algorithm. Although introducing such constructions might appear restrictive, there are many real-world examples that fit this constraint, such as in genetic studies where clusters of patients are characterized by disjoint biological structures of genes \cite{madeira2004biclustering}.  In addition, this approach has another potential advantage in some settings as the cluster of covariates can be more interpretable. Finally, the proposed algorithm and framework represent the formal basis for handling other more general biclustering configurations and analyzing complex statistics objects with biclustering techniques. 
It should be mentioned, however, that the obtention of more general biclustering configurations, e.g. using greedy optimization methods, is out of the scope of this paper.

In order to propose a biclustering RKHS algorithm capable of detecting general distributional differences in general separable Hilbert spaces, we cannot use the notions of ED and MMD directly, since random elements that constitute each subset of  $\mathcal{C}^{A,B}= \{\mathcal{C}_1,\dots, \mathcal{C}_m\}$, where $\mathcal{C}_{l}=  \{ X_i\left(\mathcal{I}_l\right): i \in \mathcal{J}_{l}\}$,  can have different dimensions. However, we can  define our biclustering problem as follows:

\begin{equation}\label{eqn:energiabiclustering4}
	\hat{\mathcal{C}}^{\hat{\mathcal{A}},\hat{\mathcal{B}}} =   
	\arg \min_{\mathcal{C}^{\mathcal{A},\mathcal{B}}= \{\mathcal{C}_1,\dots, \mathcal{C}_m\} }
	\sum_{i=1}^{m}  \frac{1}{n_{\mathcal{C}_i}} \rho_{\mathcal{I}_i}\left(\mathcal{C}_i,\mathcal{C}_i\right)=
	\arg \min_{\mathcal{C}^{\mathcal{A},\mathcal{B}}= \{\mathcal{C}_1,\dots, \mathcal{C}_m} \sum_{i=1}^{m}  \frac{1}{n_{\mathcal{C}_i}} \sum_{x\in \mathcal{C}_i} \sum_{y\in \mathcal{C}_i} \rho_{\mathcal{\mathcal{I}}_i}\left(x,y\right),
\end{equation}
where $\rho_{\mathcal{I}_i}$  denotes a metric of negative-type, related with the local kernel $k^{\mathcal{I}_i}$ previously defined in Section \ref{sec:previa}. In particular,  let $\rho_{\mathcal{I}_i}(x,y)= k^{\mathcal{I}_i}(x,x_0)+k^{\mathcal{I}_i}(y,x_0)-2k^{\mathcal{I}_i}(x,y)$ $\forall x,y\in \mathbb{H}_{\mathcal{I}_i}$  and  $x_0\in \mathbb{H}_{\mathcal{I}_i}$ be an arbitrary fixed point and where $k^{\mathcal{I}_i}(x,y)= f(\norm{x-y}_{\mathcal{I}_i})$, with $f$ being a continuous and strict monotone function. Using similar arguments to \cite{francca2020kernel}, we can establish that the problem defined in the Equation \ref{eqn:energiabiclustering4} is equivalent to:

\begin{equation}\label{eqn:biclusteringkernel}
	\hat{\mathcal{C}}^{\hat{\mathcal{A}},\hat{\mathcal{B}}}= 
	\arg \min_{\mathcal{C}^{A,B}} \sum_{i=1}^{m} \frac{1}{n_{\mathcal{C}_i}} \sum_{x\in \mathcal{C}_{i}} \norm{\phi_{\mathcal{I}_i}(x)-\hat{\mu}_{\mathcal{C}_i}}^{2}_{\mathcal{H}^{\mathcal{I}_i}},  
\end{equation}

where $\phi_{I_i}(x)= k^{\mathcal{I}_i}(\cdot,x)\in \mathcal{H}^{\mathcal{I}_i}$ and   $\hat{\mu}_{\mathcal{C}_i}= \frac{1}{n_{\mathcal{C}_i}} \sum_{x\in \mathcal{C}_{i}}\phi_{\mathcal{I}_i}(x)$.

\subsubsection{Optimization Algorithm}\label{sec:optimizacionbiclustering} 

The mains steps for solving the optimization problem defined in Eq. (\ref{eqn:biclusteringkernel}) are listed below: 

\begin{enumerate}
	\item Start by running the kernel $k$-groups algorithm introduced in the Section \ref{sec:clusteringnormal} separately at individual and covariate levels on $X^{data}$ to obtain the initial partitions    $\mathcal{B}:= \{\mathcal{J}_{1},\dots,\mathcal{J}_{k}\}$, and $\mathcal{A}:=\{\mathcal{I}_{1},\dots, \mathcal{I}_{k}\}$, respectively. If we represent the input data as a matrix $X^{input}$, we can use $X^{input}$ and its transpose $\left(X^{input}\right)^{t}$ to obtain two independent clustering for individuals and covariates, respectively.  We must note that  $X^{input}$ is not always an $n \times p$ matrix. Suppose, for instance, that each covariate is a curve that takes values in $\mathbb{H}=L^{2}(\left[0,1\right])$ and we observe the curve values in a grid on $m$ points of interval $\left[0,1\right]$. Then, $X^{input}$ will be a matrix of size $n\times (mp)$, where each subset of $m$ columns corresponds to the information about each functional covariate. 
	
	\item Given the partition $\mathcal{A}:=\{\mathcal{I}_{1},\dots, \mathcal{I}_{k}\}$ of $m$ groups of covariates, that have an associated kernel $k^{\mathcal{I}_i}$ $(
	i=1,\dots,m)$ (see Section \ref{sec:previa}),  we can build a new partition on individuals $\hat{\mathcal{B}}=\{\mathcal{J}_{1},\dots,\mathcal{J}_{k}\}$ that solves the following optimization problem:
	\begin{equation}
	\label{formula:biclusteringpaso3}
	\hat{\mathcal{C}}^{\mathcal{A},\hat{\mathcal{B}}} =     
	\arg \min_{\mathcal{C}^{\mathcal{A},\mathcal{B}}}  \sum_{i=1}^{m} \sum_{x\in \mathcal{C}_i} \norm{\phi_{\mathcal{I}_i}(x)-\hat{\mu}_{\mathcal{I}_i}}^{2}_{\mathcal{H}^{\mathcal{I}_i}}.
	\end{equation}

	The optimal solution can be obtained as detailed  
	in Section \ref{sec:clusteringnormal}. More specifically, the 
	challenge 	is to assign $x \in X^{data}$ to a new cluster using the criteria:
	\begin{align}
	j^{*} &= \arg \max_{l=1,\dots, m|l\neq j} \Delta Q^{\left(j\to l\right)}\left(x\right), \forall x\in X^{data}.
	\end{align}
	
	However, we must take into account the local structure of kernels  $k^{\mathcal{I}_i}$
	$\left(l=1,\dots,m\right)$. Thus: 
	\begin{align*}
	Q^{+}_l(x) &= Q_l+2Q_l(x)+k^{\mathcal{I}_l}(x(\mathcal{I}_l),x(\mathcal{I}_l)),\\
	Q^{-}_j(x) &= Q_j-2Q_j(x)+k^{\mathcal{I}_j}(x(\mathcal{I}_j),x(\mathcal{I}_j)),
	\end{align*}
	where $Q_{l}(x)$ and $Q_l$ are redefined as follows: 
	\begin{align*}
	Q_{l}(x) &= \sum_{y \in \mathcal{C}_{l}} k^{\mathcal{I}_l}\left(x\left(\mathcal{I}_l\right),y\right), \\
	Q_l &= \sum_{x\in \mathcal{C}_{l}} 
	\sum_{y\in \mathcal{C}_{l}} k^{\mathcal{I}_l}\left(x,y\right) \hspace{0.2cm} \left(l=1,\dots,m\right).
	\end{align*}

	Solving this problem is computationally intensive as local kernels must be recalculated according to the normalized norm defined in Section \ref{sec:previa}.

	\item Repeat step $2$ but in this case changing the role of	$\mathcal{B}:= \left \{\mathcal{J}_{1},\dots,\mathcal{J}_{m} \right\}$ by $\mathcal{A}:= \left  \{\mathcal{I}_{1},\dots, \mathcal{I}_{m} \right\}$. Bellow, we formally introduce the specific steps.  First, we define $\mathcal{C}^{t}_{l}=  \{ X_i(\mathcal{J}_l): i \in \mathcal{I}_{l}\}$

	Given the partition $\mathcal{B}:=\{\mathcal{J}_{1},\dots, \mathcal{J}_{m}\}$ of $m$ groups of individuals, that have an associated kernel $k^{\mathcal{J}_i}$ $(
	i=1,\dots,m)$ (see Section \ref{sec:previa}),  we can build a new partition on the set of  variables $\hat{\mathcal{A}}=\{\mathcal{I}_{1},\dots,\mathcal{I}_{m}\}$ that solves the following optimization problem:
	
	\begin{equation}
	\label{formula:biclusteringpaso4}
	\hat{\mathcal{C}}^{\hat{\mathcal{A}},\mathcal{B}} =     
	\arg \min_{\mathcal{C}^{\mathcal{A},\mathcal{B}}}  \sum_{i=1}^{k} \sum_{x\in \mathcal{C}^{t}_{l}} \norm{\phi_{\mathcal{J}_i}(x)-\hat{\mu}_{\mathcal{J}_i}}^{2}_{\mathcal{H}^{\mathcal{J}_i}}.
	\end{equation}

	The optimal solution can be obtained as detailed  
	in Section \ref{sec:clusteringnormal}. More specifically, the 
	challenge 	is to assign $x \in (X^{input})^{t}$ to a new cluster using the criteria:
	\begin{align}
	j^{*} &= \arg \max_{l=1,\dots, k|l\neq j} \Delta Q^{\left(j\to l\right)}\left(x\right), \forall x\in (X^{input})^{t}.
	\end{align}
	
	However, we must take into account the local structure of Kernels  $k^{\mathcal{J}_i}$
	$\left(l=1,\dots,m\right)$. Thus: 
	\begin{align*}
	Q^{+}_l(x) &= Q_l+2Q_l(x)+k^{\mathcal{J}_l}(x(\mathcal{J}_l),x(\mathcal{J}_l)),\\
	Q^{-}_j(x) &= Q_j-2Q_j(x)+k^{\mathcal{J}_j}(x(\mathcal{J}_j),x(\mathcal{J}_j)),
	\end{align*}
	where $Q_{l}(x)$ and $Q_l$ are redefined as follows: 
	\begin{align*}
	Q_{l}(x) &= \sum_{y \in \mathcal{C}^{t}_{l}} k^{\mathcal{J}_l}\left(x\left(\mathcal{J}_l\right),y\right), \\
	Q_l &= \sum_{x\in \mathcal{C}^{t}_{l}} 
	\sum_{y\in \mathcal{C}^{t}_{l}} k^{\mathcal{J}_l}(x,y) \hspace{0.2cm} \left(l=1,\dots,m\right).
	\end{align*}

    	\item Repeat alternatively steps $2$ and $3$, $T$ times, to get the succession of partitions $\left\{\mathcal{B}^t\right\}_{t=1}^{T}$ and $\left\{\mathcal{A}^t \right\}_{t=1}^{T}$, where the sub-index $t$ denotes an iteration of the algorithm. Return the solution $\mathcal{C}^{\mathcal{A}^{T},\mathcal{B}^{T}}$.

\end{enumerate}

In real-world scenarios, we know that local resolution strategies based on kernel $k$-means algorithms can strongly depend on the initial condition and get stuck in local minima. To avoid this, we apply $R$
random restarts in each phase of the algorithm as is common practice in cognate algorithms. 

Bellow, we present some results that provide theoretical guarantees the optimization strategy of this Section \ref{sec:optimizacionbiclustering}. In addition, we offer the computational complexity cost of our procedure.

\begin{theorem}
	AKKB biclustering algorithm converges in a finite number of steps.
\end{theorem}

\begin{theorem}
	The complexity of AKKB biclustering algorithm is $\mathcal{O}\left(k^{2}n^{2}p^{2}\right)$, where $k$ is the number of clusters, $n$ is the number of data points, and $p$ the number of covariates. 
\end{theorem}

\subsubsection{Statistical Theory}\label{sec:theory}

This section introduces some mathematical properties of our biclustering proposal as well as consistency results. Rigorous mathematical proofs are relegated to Supplementary Material.   

Before starting, we present some statistical learning concepts related to clustering analysis. For this aim, we adopt as a  reference  the notation introduced in \cite{biau2008performance}. 
Specifically, 
let us consider the following empirical risk function associated with our biclustering problem:

\begin{equation}
 \label{formula:demostracioninicial}
	\mathcal{W}_n \left(c_n, \mathcal{A}= \left\{\mathcal{I}_{1},\dots, \mathcal{I}_{m}\right\}, X^{data},\mu_n\right)
	\\= \min_{\mathcal{C}^{\mathcal{A},\mathcal{B}}= \left(\mathcal{C}_1,\dots, \mathcal{C}_m\right)}  \frac{1}{n} \sum_{i=1}^{m}  \sum_{x\in \mathcal{C}_i}^{} \norm{\phi_{\mathcal{I}_i}\left(x\right)- c_n\left(I_i\right)}_{\mathcal{H}^{\mathcal{I}_i}}^{2},
\end{equation}
where $\mu_n$ is the empirical measure of $X^{data}$,  $\mu_n= \frac{1}{n}\sum_{i=1}^{n} \delta_{X_i}$, and, for any $\mathcal{I}_{i}\in \mathcal{P}\left(\{1,\dots,p\}\right)$, 

$c_n:  \mathcal{I}_{i}\in \mathcal{P}\left(\{1,\dots,p\}\right)   \to  c_n\left(\mathcal{I}_{i}\right)  \in  \mathcal{H}^{\mathcal{I}_i}$ is a set-valued function in an appropriate $RKHS$, denoting by $\mathcal{P}\left(\left\{1,\dots,p\right\}\right)$ the power set of covariate indexes  $\left\{1,\dots,p\right\}$.


The problem defined in Eq.  (\ref{formula:demostracioninicial}) finds the partition at the individual level that minimizes the empirical risk function, from a fixed set-valued function $c_n$ and the cluster covariates $\mathcal{A}$.

Consider now the population counterpart:

\begin{equation}
	\label{formula:demostracion2}
	\mathcal{W}\left(c, \mathcal{A},\mu\right) = \int \min_{1\leq j \leq k} \norm{\phi_{\mathcal{I}_j}\left(x\left(\mathcal{I}_j\right)\right)- c\left(I_j\right)}^{2}_{\mathcal{H^{\mathcal{I}_j}}} d\mu\left(x\right),
\end{equation}

where $\mu$ denotes the measure of the random variable $X$, $x\left(\mathcal{I}_i\right)= \left(x^i\right)_{i\in \mathcal{I}_i}$, and $c:\mathcal{I}_{i}\in \mathcal{P}\left(\{1,\dots,p\}\right)   \to  c\left(\mathcal{I}_{i}\right)  \in  \mathcal{H}^{\mathcal{I}_i}$ is a fixed set-valued function in a RKHS. 
Then, the optimal clustering risk can be defined as:
\begin{equation}
\label{formula:demostracion3}
\mathcal{W}^{*}\left(\mu\right)= \inf_{\mathcal{A}} \inf_{c} \mathcal{W}\left(\mathcal{A}, c,\mu\right).
\end{equation}

We  note that for each $\mathcal{I}_i$ fixed, the infimum of $c\left(\mathcal{I}_i\right)$ is found in the mean, due to the Euclidean geometry of the quadratic problem.

Following the theoretical framework established in \cite{biau2008performance} about clustering in Hilbert spaces, we say that the partition of the set of covariates $\mathcal{A}_n$ and the cluster centers $c_n$ are $\delta_n$ minimizer of the empirical clustering risk if:
\begin{equation}
\mathcal{W}_n\left(c_n, \mathcal{A}_n, X^{data},\mu_n\right)\leq \mathcal{W}^{*}\left(\mu_n\right)+\delta_n
\end{equation}

To establish consistency results about empirical biclustering risk over $RKHS$, we assume that $\delta_n\to 0$ when $n\to \infty$. In practice, this means that $\delta_n$-minimizers of the empirical clustering risk converge to the optimal risk. We are ready to establish the next result.

\begin{theorem}
	Let $A_n$ and $c_n$ be $\delta_n$-minimizers of the empirical clustering risk such that when $n\to \infty$, $\delta_n\to 0$. Then:
	\begin{enumerate}
		\item $\lim_{n\to \infty} \mathcal{W}\left(c_n, \mathcal{A}_n, X^{data},\mu_n\right)= \mathcal{W}^{*}\left(\mu\right)$ with probability one.
		\item $\lim_{n\to \infty} \mathbb{E}\left(\mathcal{W} \left(A_n, c_n, \mu_n\right)\right)=  \mathcal{W}^{*}\left(\mu\right)$
	\end{enumerate}
\end{theorem}

Our previous theorem shows that the resolution of empirical problems with finite samples have theoretical guarantees to converge to the population biclustering risk.

\section{Experiments}\label{sec:results}

$AKKB$\footnote{An \texttt{R} package is available at \url{https://https://github.com/kbiclustering/kernel.biclustering} to reproduce the results.} performance has been compared against the following state-of-the-art algorithms, using synthetics and real datasets:

\begin{itemize}
	\item Alternating $k$-means biclustering ($AKM$)\footnote{\texttt{R} package \texttt{akmbiclust}} \cite{fraiman2020biclustering}: This algorithm finds local minimum by alternating the use of an adapted version of the $k$-means clustering algorithm between columns and rows, and is a particular case of $AKKB$.
	
	\item Profile likelihood biclustering ($PL$)\footnote{\texttt{R} package \texttt{biclustpl}} \cite{flynn2020profile}: This method is based on 
	the maximal profile likelihood using as a reference exponential family distributions. In our experiments, the algorithm was run using the Gaussian distribution family.
	
	\item Sparse biclustering ($SBC$)\footnote{\texttt{R} package \texttt{sparseBC}} \cite{tan2014sparse}: This algorithm  assumes that data entries are Gaussian with a bicluster-specific mean and equal variance that maximize the $L_1$-penalized log-likelihood to obtain sparse biclusters. An optimal $\lambda-$regularization parameter is selected according to the Bayesian information criterion (BIC).
\end{itemize}

As in \cite{fraiman2020biclustering}, we selected the former algorithms in this comparison since:
\begin{enumerate}
	\item Each datum and covariate are grouped in a unique cluster. Thus, biclustering solutions should produce non-overlapping clusters.
	\item The selected biclustering methods allow to explicitly specify the number of clusters at both individual and covariate levels.
\end{enumerate}
 
Along with this paper, and according to the definition established in Section \ref{sec:previa}, $AKKB$ was fitted using a Gaussian kernel. In the kernel definition, we consider two different kernel bandwidths, at individual and covariate levels, that we denote by $\sigma_{data}$ and $\sigma_{variables}$, respectively. In both cases, the median heuristic was used to obtain an initial estimation of this kernel parameter. Supplementary Material provides additional results in which we analyze the impact of the bandwidth selection at both levels. 

We  note that the median heuristic is one of the most popular criteria to select the bandwidth parameter with Gaussian kernels. In \cite{ramdas2015decreasing, garreau2017large}, the authors show that the application of median heuristic  strategy can maximize MMD power in different simulation scenarios. However, we do not have any guarantee about its performance in the biclustering set-ups.

Finally, for every algorithm we set $R=100$ starts.

.

\subsection{Simulated data}
One thousand simulations of different random processes $X_{ij}$ $\left(i=1,\dots,200; j=1,\dots,200\right)$ were performed.
We assume that the random process $X_{ij}$ defines two latent clusters in a block matrix $2 \times 2$.  The matrix $A$ is defined as: 
\begin{align*}
A = \left(X_{ij}\right) =  \left(
\begin{array}{c|c}
A_{11} & A_{12} \\ \hline
A_{21} & A_{22}
\end{array}
\right)
\end{align*}
where:
\begin{align*}
A_{11} &= \left(X_{ij}\right)^{i=1,\dots,100}_{j=1,\dots,100} \;\;\; & 
A_{12} &= \left(X_{ij}\right)^{i=1,\dots,100}_{j=101,\dots,200} \\
A_{21} &= \left(X_{ij}\right)^{i=101,\dots,200}_{j=1,\dots,100} \;\;\;  &
A_{22} &= \left(X_{ij}\right)^{i=101,\dots 200}_{j=101,\dots,200}
\end{align*}

The following scenarios were considered to illustrate the theoretical advantages of $AKKB$:

\begin{enumerate}
	\item Normally distributed data such that clustering differences are in variance, whereas all data have the same mean:
	\begin{align*}
	A_{11}  &\stackrel{}{\sim} N\left(0,2\right), \\ 
	A_{12}&=A_{21} = A_{22}\stackrel{}{\sim} N\left(0,1\right).
	\end{align*}
	
	\item Data distributed according to a uniform distribution, in which each block of the random matrix $A$ has the same mean, but there are differences in the remaining moments:
	\begin{align*}
	A_{11} &\stackrel{}{\sim}  Unif\left(0.3,0.7\right), \\  
	A_{12} &= A_{21}= A_{22}\stackrel{}{\sim} Unif\left(0,1\right).
	\end{align*}
	
	\item A random matrix $A$ distributed according to uniform, standard Gaussian, and truncated Gaussian distributions, with the same first three moments. In particular:
	\begin{align*}
	A_{11} & \stackrel{}{\sim}  Unif\left(-\sqrt{3},\sqrt{3}\right),  
	A_{22} \stackrel{}{\sim} N\left(0,1\right), \\
	A_{12} & = A_{21} \stackrel{}{\sim} N_{truncated}\left(0,1,-1.8,1.8\right) \\
	&\;\;\;\;\;\;\;\;\;\;\;\;+ Unif\left(-0.5,0.5\right)
	\end{align*}
\end{enumerate}

The average accuracy reached by each algorithm in the three scenarios is listed in Table \ref{section:table:tabla1}. We can see that $AKKB$ outperforms the other algorithms in all three scenarios analyzed. As expected, $SBC$ and $PL$ present worse performance when the differences between clusters are due to changes in moments that go beyond the mean.

\begin{table}[tb!]
	\caption{Average accuracy in the three synthetic scenarios with $1000$ trials.}
	\label{section:table:tabla1}
	\vskip 0.15in
	\begin{center}
		\begin{small}
			\begin{sc}
				\begin{tabular}{ccccc}
					\toprule
					Scenario & $AKKB$  & $SBC$  & $PL$ &  $AKM$\\
					\midrule
					$1$    & $1$ & $0.52$ & $0.535$  & $1 $ \\
					$2$ & $0.92$ & $0.54$  & $0.534$ & $0.78$ \\
					$3$    & $0.91$ & $0.515$ & $0.505$  & $0.83$ \\
					\bottomrule
				\end{tabular}
			\end{sc}
		\end{small}
	\end{center}
	\vskip -0.1in
\end{table}

\subsection{Genetic expression data-sets}

We also ran comparisons between $AKKB$ and the baselines on real biomedical examples, specifically on three genetic expression datasets from the \textit{Schliep lab} repository\footnote{Available at  \url{https://schlieplab.org/Static/Supplements/CompCancer/datasets.htm}} in which true structures are annotated.
The results are introduced in Table \ref{section:table:tabla2}.

\begin{table}[tb!]
	\caption{Average accuracy in three real dataset, where $N$ denotes the sample size and $P$ the number of covariates.}
	\label{section:table:tabla2}
	\vskip 0.15in
	\begin{center}
		\scalebox{0.8}{
			\begin{small}
				\begin{sc}
					
					\begin{tabular}{lcccc}
						\toprule
						Data set \text{(n,p)} & $AKKB$  & $SBC$  & $PL$ &  $AKM$\\
						\midrule
						West-2001 $(49,1198)$     & $0.85$ & $0.51$ & $0.81$  & $0.81$ \\
						Chowdary-2006 $(42,182)$     & $0.98$ & $0.65$  & $0.65$ & $0.96$ \\
						Armstrong-2002-v1  $(72,1081)$ & $0.75$ & $0.90$ & $0.61$  & $0.76$  \\
						\bottomrule
					\end{tabular}
				\end{sc}
		\end{small}}
	\end{center}
	\vskip -0.1in
\end{table}

As we can see, $AKKB$ obtains competitive results in all three datasets. However, $SBC$ outperforms $AKKB$ in Armstrong-2002-v1, which can indicate that an underlying sparse structure may often be present in the biological data. This might indicate that applying variable selection or dimension reduction before AKKB could increase the performance of our model. 

\subsection{Functional data diabetes example}

Functional data analysis \cite{ramsay2007applied} is a research area that has received substantial attention in recent years from the statistical community to get new insight into the analysis of objects that vary in a continuum as a temporal curve of a physiological process.

In this Section, to show the versatility of our biclustering algorithm, we will analyze a functional example. We have different functional profiles obtained from the information available from a continuous glucose monitor.

More specifically, we will use data from a case-control study \cite{weinstock2016risk} that aims to analyze risk factors associated with severe hypoglycemia in patients with type I diabetes. On each day with CGM information we observe $288$ spaced observations every $5$ min that we denote by the index $j$ $\left(j=1,\dots,288\right)$, and we define the related  time $t_j$= $5j$.

For each participant $i$, $i=1,\dots,n$, we have available $n_i$ days of continuous glucose monitoring, $n_i \in  \left\{10,\dots, 16\right\}$. Let us consider the random process $X_{id}(t_j)$, $i=1,\dots, n$; $d=1, \dots, n_i$; $j=1, \dots, 144$, which models the glucose concentration of subject $i$, on day $d$, and for time $t_j$. Also, for each subject $i$, let us consider the following random processes related with the morning and afternoon (the first $12$ hours of the day and the rest-720 minutes):
\begin{itemize}
    \item $X^{mean,morning}_i\left(t_j\right)= \frac{1}{n_i} \sum_{d=1}^{n_i} X_{id}\left(t_j\right)$,
    \item $X^{sd, morning}_i\left(t_j\right)=  \sqrt{ \frac{1}{n_i} \sum_{d=1}^{n_i} (X_{id}\left(t_j)\right)-X^{mean}_i\left(t_j\right)^{2}}$,
    \item $X^{CV, morning}_i\left(t_j\right)= \frac{X^{mean}_i(t_j)}{X^{sd}_i\left(t_j\right)}$,
    \item $X^{Skew, morning}_i\left(t_j\right)=  \sum_{d=1}^{n_i} \left[\frac{ (X_{id}(t_j)-X^{mean}_i\left(t_j\right))}{X^{sd}_i(t_j)}\right]^{3}$,
    \item $X^{curtos, morning}_i\left(t_j\right)= \sum_{d=1}^{n_i}[\frac{ (X_{id}(t_j)-X^{mean}_i\left(t_j\right))}{X^{sd}_i\left(t_j\right)}]^{4}$,
    \item $X^{mean,afternoon}_i\left(720+t_j\right)= \frac{1}{n_i} \sum_{d=1}^{n_i} X_{id}\left(720+t_j\right)$,
    \item $X^{sd,afternoon}_i\left(720+t_j\right)=  \sqrt{ \frac{1}{n_i} \sum_{d=1}^{n_i} (X_{id}\left(720+t_j\right)-X^{mean}_i\left(720+t_j\right)^{2}}$,
    \item $X^{CV,afternoon}_i\left(720+t_j\right)= \sum_{d=1}^{n_i} \frac{X^{mean}_i\left(720+t_j\right)}{X^{sd}_i\left(720+t_j\right)}$,
    \item $X^{Skew,afternoon}_i\left(720+t_j\right)= \sum_{d=1}^{n_i} \left[\frac{ (X_{id}\left(720+t_j\right)-X^{mean}_i\left(720+t_j\right))}{X^{sd}_i\left(720+t_j\right)}\right]^{3}$,
    \item $X^{curtos,afternoon}_i\left(720+t_j\right)= \sum_{d=1}^{n_i} \left[\frac{ \left(X_{id}\left(720+t_j\right)-X^{mean}_i\left(720+t_j\right)\right)}{X^{sd}_i\left(720+t_j\right)}\right]^{4}$,
\end{itemize}
where $j=1,\dots, 144$. In this example, we assume that each functional covariate analyzed takes values in the space $\mathcal{H}= L^{2}([0,720])$.

Figures \ref{fig:grafc1}-\ref{fig:graficonueva} show the results of running the AKKB algorithm ($m=2$) clusters. In the graph, we can observe those patients with a more stable glycemic control (color Red) and that their average daily glycemic values characterize them in the morning and afternoon. In the other group (color Blue), we see patients with higher volatility, whose covariates are those related to variables measuring different modes of data variability. 
The figure also shows the values of the  Glycosylated hemoglobin (A1C)  variable and Fasting Plasma Glucose (FPG) variable -- the primary variables to diagnose and control Diabetes Mellitus \cite{kottgen2007reduced,S14}, while each color represents the group that belongs to each patient. These graphics suggest that our functional biclustering results are related to the quality of glycemic control. In general, incorporating functional information of CGM data provides new insight into diabetes data analysis (see for example \cite{matabuena2021glucodensities, gaynanova2020modeling}).
\begin{figure}[tb!]
		\centering
		\includegraphics[scale=.65]{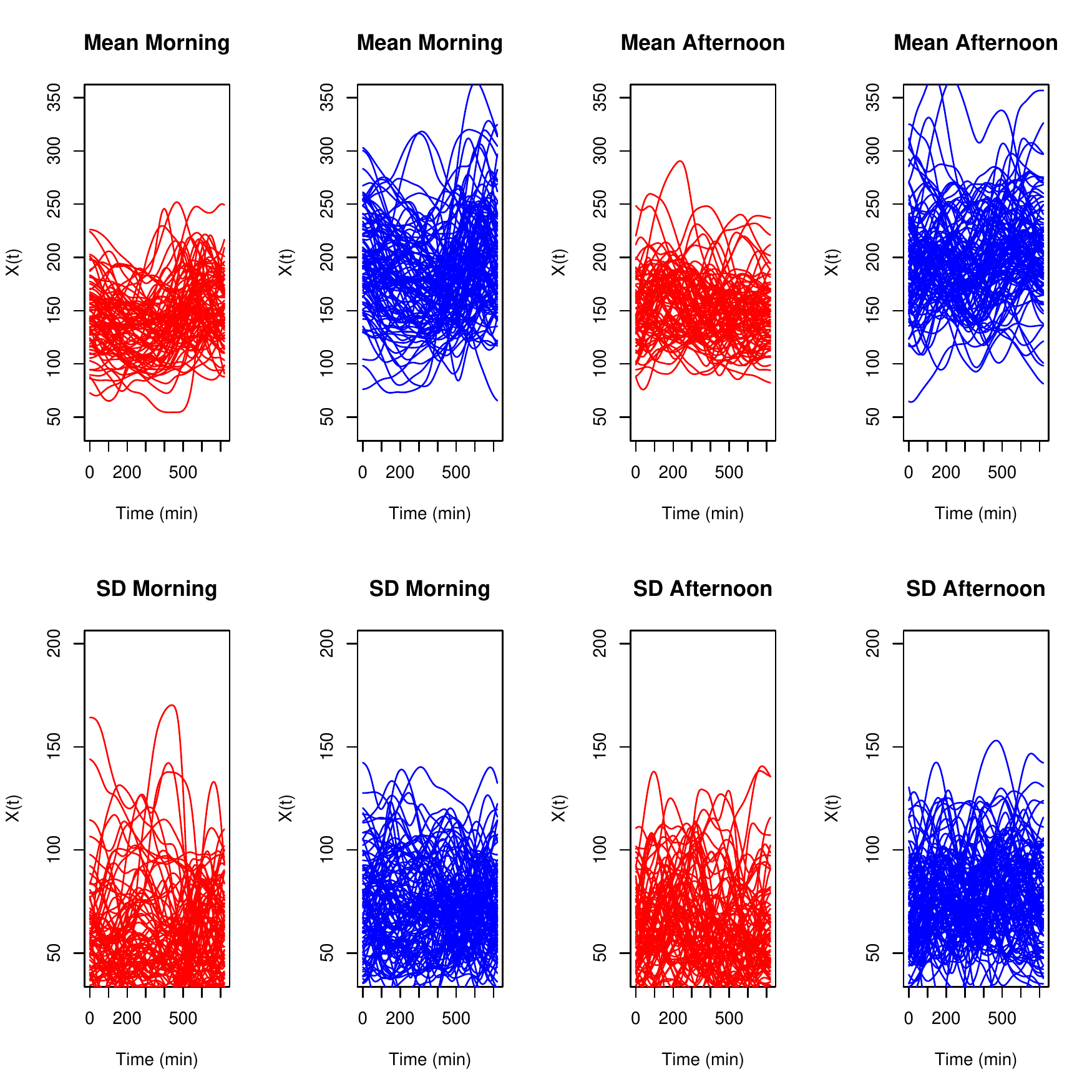}
\caption{ Morning and afternoon functional profiles on the covariates. $X^{mean}$, $X^{sd}$ .  Red (patients with reasonable glycaemic control). Blue (patients with worse glycaemic control).}
\label{fig:grafc1}
\end{figure}

\begin{figure}[tb!]
		\centering
		\includegraphics[scale=.65]{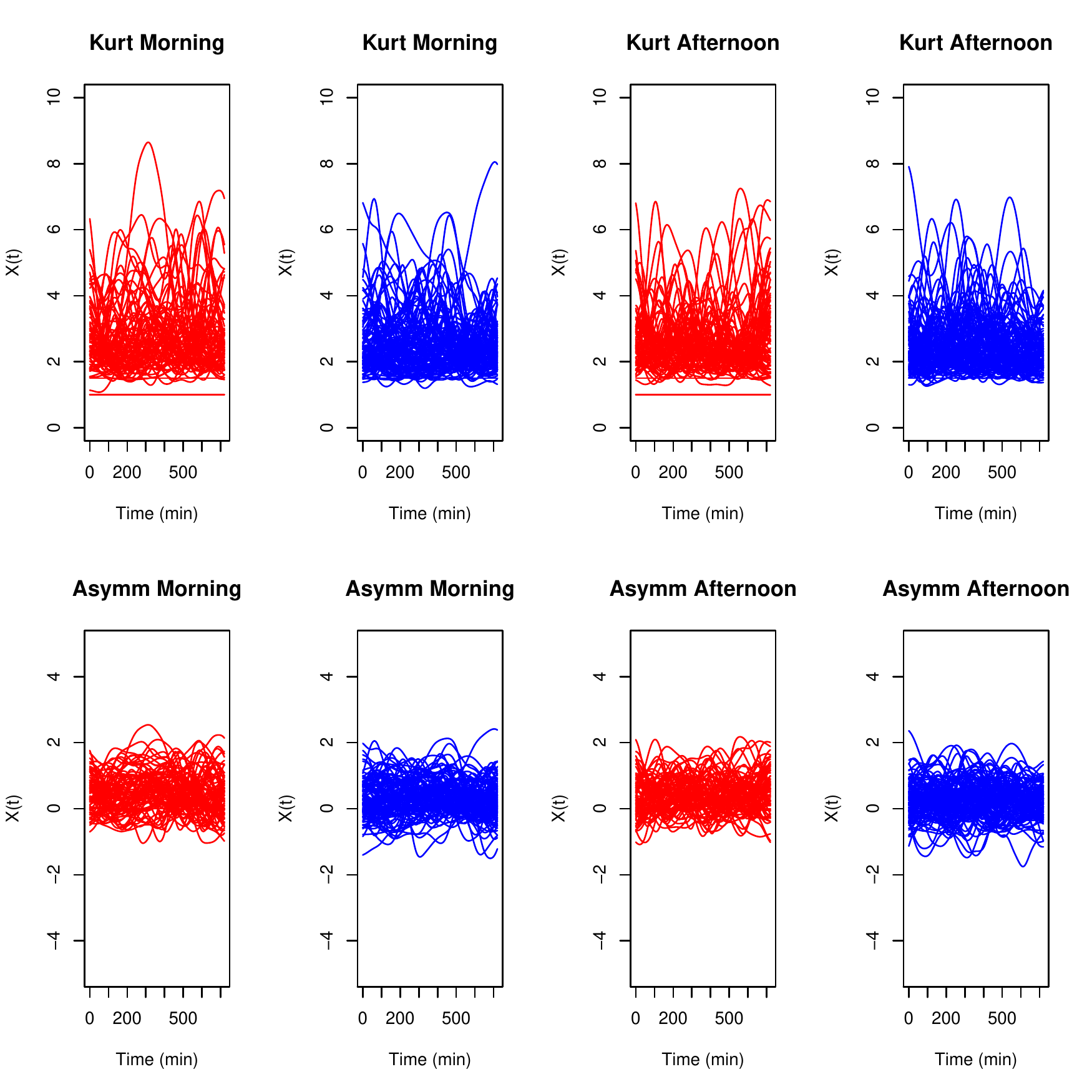}
\caption{ Morning and afternoon functional profiles on the covariates $X^{Skew}$, $X^{curtos}$.  Red (patients with reasonable glycaemic control). Blue (patients with worse glycaemic control).}
\label{fig:grafc2}
\end{figure}

\begin{figure}[tb!]
	    \centering
	    \includegraphics[scale=.6]{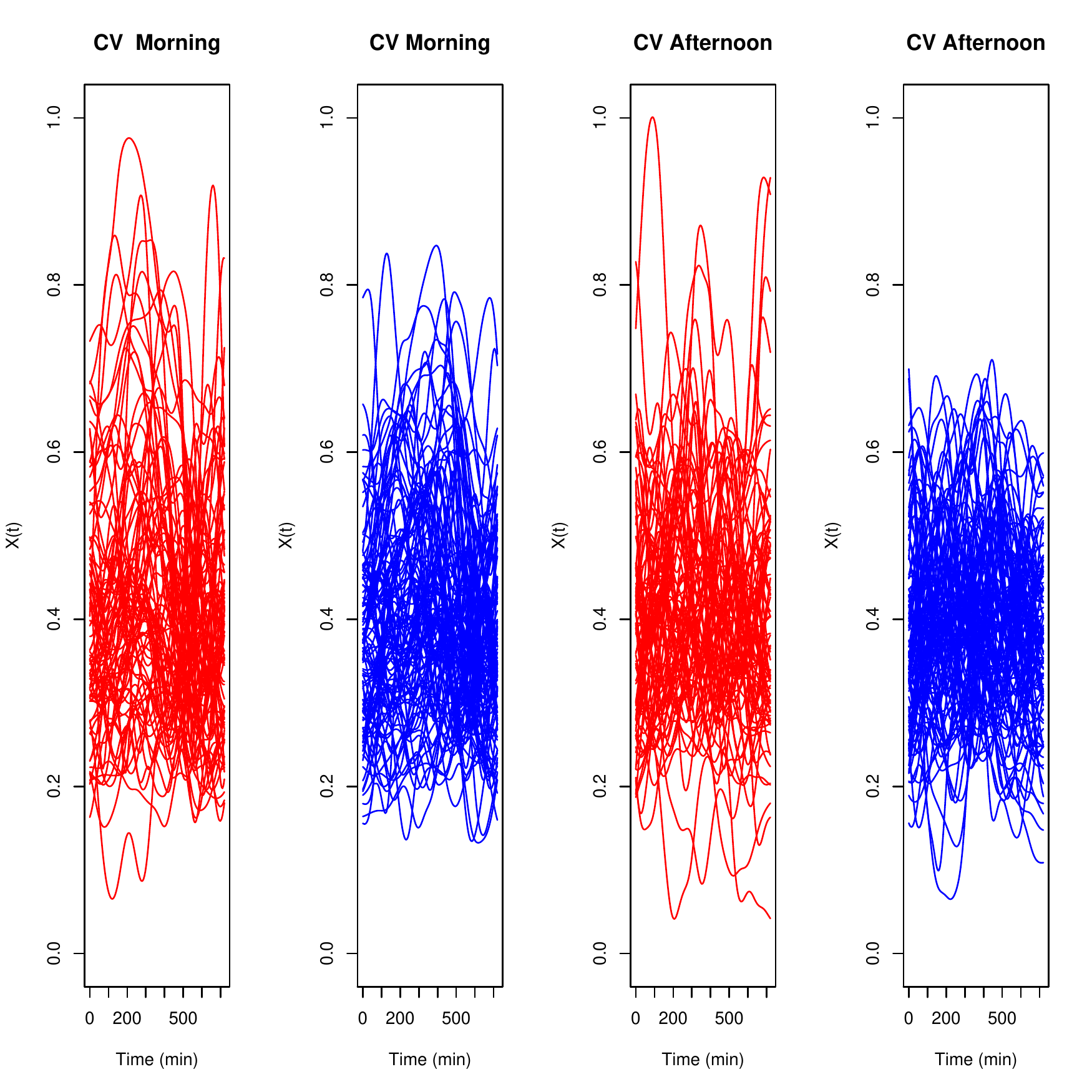}
\caption{ Morning and afternoon functional profiles on the covariate $X^{CV}$.  Red (patients with reasonable glycaemic control). Blue (patients with worse glycaemic control).}
\label{fig:grafc3}
\end{figure}

\begin{figure}[ht!]
	\centering
	\includegraphics[width=0.7\linewidth]{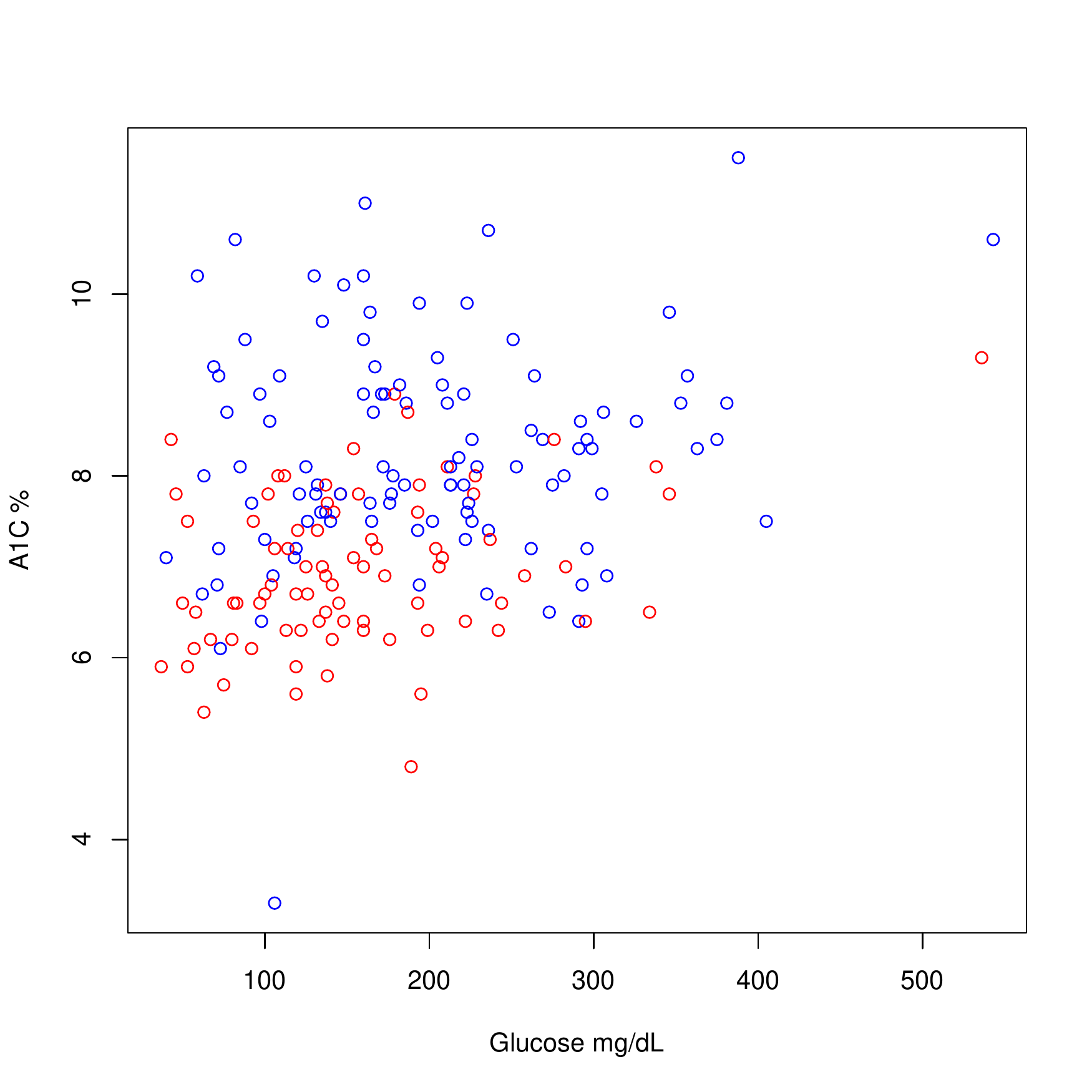}
	\caption{Bidimensional plot:  Glycosylated hemoglobin (A1C) vs. Fasting Plasma Glucose (FPG).  Red (patients with reasonable glycaemic control). Blue (patients with worse glycaemic control). }
	\label{fig:graficonueva}
\end{figure}


\section{Discussion and future work}\label{sec:discusion}

In this work, we have proposed a new formulation of the biclustering problem in separable Hilbert spaces inspired by the energy distance and the MMD, together with an efficient resolution algorithm based on the alternating application of the kernel $k$-means algorithm on the data and covariate sets. Empirical results with vector data show the advantages of the proposed method over  existing techniques in some settings, illustrating its usefulness in more general contexts such as functional data analysis. This includes settings such as the functional data diabetes example, where very few methodological approaches exist.

With the choice of the appropriate kernel function, we can treat other biclustering algorithms such as $k$-means as a particular case of the proposed framework or formulate new graph-partitioning algorithms based on the methodology discussed here only by modifying the kernel \cite{francca2020kernel}. 

Our algorithm can also  be combined with dimensionality reduction techniques such as principal component analysis or random projections. 
The application of these techniques is  recommended when the data has a sparse structure. 
In a comprehensive analysis, we have seen that the application of reduction dimension techniques improve the model performance. 

Variable selection can be an important landmark in many high-dimensional problems. We know that kernels,
when $p$ is growing with $n \to \infty$, can drastically deteriorate in their performance \cite{ramdas2015decreasing}.

Biclustering algorithms have been massively used in biology problems. However, their applications are not exclusive to this area. They are also commonly used in other domains such as consumer analysis, sports, or text mining \cite{kasim2016applied}, and the developed framework can also find use-cases therein.

In order to solve a biclustering algorithm in the RKHS efficiently, we have introduced the constraint that the clusters formed at the subject and covariate level are mutually exclusive. At first glance, it may appear as a hard constraint. However, the experimental results suggest a good model performance in analyzing several gene expression data.
This could be explained by the fact that the hypothesis is not restrictive in these domains, where genetic regions determine specific biological functions. In addition, a potential advantage of our procedure is the increase of interpretability with the inclusion of this structural constraint into the model.

From the methodological point of view, it would be interesting to explore the possibility of solving other more general formulations of biclustering using kernels for complex objects that can be defined in separable Hilbert spaces. From the mathematical perspective, it is not challenging to formalize these algorithms from the framework discussed here, but, as we already mentioned, they are computationally prohibitive. Thus, heuristic optimization techniques would be required. The selection of a kernel together with its parameters is a crucial problem in the performance of the methods discussed here. It remains broadly open, although it is noteworthy to mention recent advances in this area last few years \cite{liu2020learning,li2019implicit, jitkrittum2016}. 

Finally, a fundamental problem in clustering analysis is establishing the significance of the obtained clusters. In general, it is not straightforward, and to do it correctly, we would have to take into account the number of clusters built through the selective process (the problem of post-selection inference). To the best of our knowledge, in clustering analysis only the following references \cite{gao2020selective, chen2022selective} tackles this problem in the context of hierarchical clustering algorithms. However, following recent progress in the MMD tests kernel choice, or more specifically with U-V incomplete statistics \cite{lim2020more}, in the future, we can develop a specific procedure for choosing the kernel in the kernel $k$-groups algorithm or, more specifically, in the proposed $AKKB$ biclustering algorithm.

\begin{appendix}

\section{Mathematical Proofs}

\subsection{Computational theory}

\begin{theorem}
	AKKB biclustering algorithm converges in a finite number of steps.
\end{theorem}

\begin{proof}
	The key point of all optimization strategy is consider alternatively the following optimization  problem so at individual and covariate level:
    \begin{align}\label{eqn:crear}
    j^{*} &= \arg \max_{l=1,\dots, k|l\neq j} \Delta Q^{\left(j\to l\right)}\left(x\right), \forall x\in X^{data}.
    \end{align}
    By construction, since the resolution of Eq. (\ref{eqn:crear}) sure that the cost function $Q$  is monotonically increasing at each iteration, and there are a finite number of distinct cluster assignments, the algorithm converges in a finite number of steps.
\end{proof}

\begin{theorem}
	The complexity of AKKB biclustering algorithm is $\mathcal{O}\left(k^{2}n^{2}p^{2}\right)$, where $k$ is the number of clusters, $n$ is the number of data points. 
\end{theorem}

\begin{proof}
    According to \cite{francca2020kernel}, the complexity of kernel k-group algorithm is the order $\mathcal{O}\left(kn^{2}\right)$. Since our algorithm consists of independently and alternately applying the kernel $k$-group algorithm at the covariate and at the individual level, the overall complexity is $\mathcal{O}\left(k^{2}n^{4}\right)= \mathcal{O}\left(kp^{2}\right)\times \mathcal{O}\left(kn^{2}\right)$. 
\end{proof}

\subsection{Statistical  theory}

We start to introduce the  $L_2$-Wasserstein distance in RKHS following the strategy established in  \cite{zhang2019optimal}.

Let the input space $\left(\mathbb{H}^p,\mathcal{B}_{\mathbb{H}^p}\right)$ be a measurable space with a Borel $\sigma$-algebra $\mathcal{B}_{\mathbb{H}^p}$, and $\mathbb{P}$ be the set of probability measures on   $\mathbb{H}^p$ with finite-second order moments. Let $k: \mathbb{H}^p\times \mathbb{H}^p\to \mathbb{R}^{+}$ be a positive definite kernel, and $\mathbb{H}_{k}$ be the  RKHS generated by $k$. Let $\phi:\mathbb{H}^p\to \mathbb{H}_{k}$ the corresponding embedding. For any $\mu \in \mathbb{P}$, let $\phi_{\sharp}\mu$ be the push-forward measure of $\mu$\footnote{Given a probability measure $\mu$ on the input space, mapping the data through the implicit map $\phi$, we are interested in the data distribution in RKHS. Such distribution is called the push-forward measure denoted as  $\phi_{\sharp}\mu$, satisfying that for any subset $A$ in RKHS, $\phi_{\sharp}\mu\left(A\right)=\mu\left(\phi^{-1}\left(A\right)\right)$.}.

Given $\mu, \eta\in \mathbb{P}$, the Wasserstein distance between push-forward measures  $\phi_{\sharp}\mu$ and  $\phi_{\sharp}\eta$ on $\mathbb{H}_{k}$ is defined as:
\begin{equation}\label{eqn:transporte1}
	\gamma\left(\phi_{\sharp}\mu, \phi_{\sharp}\eta \right)=\left[\inf_{\substack{\pi_{k} \in
			\Pi\left(\phi_{\sharp}\mu,\phi_{\sharp}\eta\right)}} 
			\int_{}\norm{x-y}_{k}^{2} d\pi_{k}\left(x,y\right)\right]^{1/2}, 	
\end{equation}
where $\Pi\left(\phi_{\sharp}\mu,\phi_{\sharp}\eta \right)$ is the set of joint probability measures on $\mathbb{H}^p\times  \mathbb{H}^p$, with marginals $\phi_{\sharp}\mu$ and $\phi_{\sharp}\eta$.

Bellow, we enunciate an equivalent and computable formulation, which is fully determined by the kernel function. 
\begin{theorem}
	Let $\left(\mathbb{H}^p,\mathcal{B}_{\mathbb{H}^p}\right)$ a Borel space, and let the reproducing kernel $k$ be measurable. Given $\mu, \eta\in \mathbb{P}$, we write:
	
	\begin{equation}\label{eqn:transporte2}
		\gamma^{k}\left(\mu, \eta\right)= \left[\inf_{ \pi  \in \Pi\left(\mu,\eta\right)} \int_{} \norm{\phi\left(x\right)-\phi\left(y\right)}^{2} d\pi\left(x,y\right)\right]^{1/2}. 	
	\end{equation}
\end{theorem}
\noindent where $\norm{\phi\left(x\right)-\phi\left(y\right)}^{2}= k\left(x,y\right)+k\left(y,y\right)-2k\left(x,y\right).$ Then:  

\begin{itemize}
	\item	$\gamma^{k}\left(\mu, \eta\right)= \gamma \left(\phi_{\sharp}\mu, \phi_{\sharp}\eta \right).$
	\item  If $\pi^{*}$ is the minimized of Eq. (\ref{eqn:transporte2}), then $\left(\phi, \phi\right)_{\sharp} \pi{*}$ is a minimizer of  Eq. (\ref{eqn:transporte1}), where $\left(\phi, \phi\right): \mathbb{H}^p\times \mathbb{H}^p \to  \mathbb{H}_{k}\times \mathbb{H}_{k}$, is defined as  $\left(\phi,\phi\right)\left(x,y\right)= \left(\phi\left(x\right),\phi\left(y\right)\right)$.
	
\end{itemize}

\begin{theorem}
	Let $A_n$ and $c_n$ be $\delta_n$-minimizers of the empirical clustering risk such that when $n\to \infty$, $\delta_n\to 0$. Then:
	\begin{enumerate}
		\item $\lim_{n\to \infty} \mathcal{W}\left(c_n, \mathcal{A}_n, X^{data},\mu_n\right)= \mathcal{W}^{*}\left(\mu\right)$ with probability one.
		\item $\lim_{n\to \infty} \mathbb{E}\left(\mathcal{W}\left(A_n, c_n, \mu_n\right)\right)=  \mathcal{W}^{*}\left(\mu\right)$.
	\end{enumerate}
\end{theorem}

Now, we are in conditions to introduce the Lemmas that compose the main proof of this paper.

\begin{lemma}
	\label{lemma-1}
	For any column partition $\mathcal{A}=\left\{\mathcal{I}_1,\dots, \mathcal{I}_m \right\}$, and set-value function $c$, we have:
	\begin{equation}
		\left|\sqrt{\mathcal{W}\left(\mathcal{A},c,\mu\right)}-\sqrt{\mathcal{W}\left(\mathcal{A},c,\eta\right)}\right|\leq \gamma^{k}\left(\mu, \eta\right).
	\end{equation}
\end{lemma}
\begin{proof} 
	Let $X \stackrel{}{\sim} \mu$ and $Y \stackrel{}{\sim} \eta$,  achieve the infimum defining  $\gamma^{k}\left(\mu, \eta\right)$. Then:
	\begin{align}
\sqrt{\mathcal{W}\left(\mathcal{A},c,\mu\right)} = 
		\left( \int \min_{1\leq j \leq m }   
	\norm{\phi\left(x\left(\mathcal{I}_j\right)\right)- c\left(\mathcal{I}_j\right)}_{k_{\mathcal{I}^j}}^{2}  d\mu\left(x\right)  \right)^{\frac{1}{2}}
	\end{align}		
	\noindent which implies one direction of the inequality enunciate. The other direction can be proved analogously.
\end{proof}

%

\begin{lemma}
	\label{lemma-2}
	Let $\mathcal{A}_n$ and $c_n$ be a $\delta_n$ minimizers of the empirical clustering risk. Then:
	\begin{equation*}
		\sqrt{	\mathcal{W}\left(\mathcal{A},c_n, \mu\right)}- \sqrt{\inf_{\mathcal{A}} \inf_{c} 	\mathcal{W}\left(\mathcal{A},c,\mu \right)} \leq 2\gamma^{k}\left(\mu, \mu_n\right)+\sqrt{\delta_n}
	\end{equation*}
\end{lemma}
\begin{proof}
		Let $\epsilon>0$. In virtue of supremum axiom of real numbers, let $\mathcal{A}^{*}$ and $c^{*}$ be the elements satisfying:
		\begin{equation}
			\inf_{I} \inf_{c} \mathcal{W}\left(\mathcal{A},c,\mu \right) \leq \mathcal{W}\left(\mathcal{A}^{*}, c^{*}, 	\mu\right) \leq \inf_{\mathcal{A}} \inf_{c} \mathcal{W}\left(\mathcal{A},c,\mu \right) +\epsilon.
		\end{equation}
		For any $x\in \mathbb{R}$, we define $\left(x\right)_{+}= \max \left(x,0\right)$. We have:
		\begin{flalign*}
			\sqrt{	\mathcal{W}\left(\mathcal{A}_n,c_n, \mu\right)} - \sqrt{\inf_{\mathcal{A}} \inf_{c} \mathcal{W}\left(\mathcal{A},c,\mu \right)} 
			& \leq   	
			\sqrt{	\mathcal{W}\left(\mathcal{A}_n,c_n, \mu\right)} - \sqrt{\left[\mathcal{W}\left(\mathcal{A}^{*}, c^{*}, \mu\right) -\epsilon\right]_{+}} \\
			& \leq 	\sqrt{\mathcal{W}\left(\mathcal{A}_n, c_n, \mu\right)} - \sqrt{\mathcal{W}\left(\mathcal{A}^{*}, c^{*}, \mu\right)} + \sqrt{\epsilon} \\
		     & =\sqrt{\mathcal{W}\left(\mathcal{A}_n, c_n, \mu\right)} - \sqrt{\mathcal{W}\left(\mathcal{A}_n,c_n,\mu_n\right)} \\ 
		     & \quad + \sqrt{\mathcal{W}\left(\mathcal{A}_n,c_n,\mu_n\right)} - \sqrt{\mathcal{W}\left(\mathcal{A}{*}, c^{*}, \mu\right)}+  \sqrt{\epsilon} \\ 
		    & \leq 	\sqrt{\mathcal{W}\left(\mathcal{A}_n, c_n, \mu\right)}- \sqrt{\mathcal{W}\left(\mathcal{A}_n,c_n,\mu_n\right)} \\
		    & \quad + \sqrt{\mathcal{W}\left(\mathcal{A}^{*},c^{*},\mu_n\right)} -   \sqrt{\mathcal{W}\left(\mathcal{A}^{*}, c^{*}, \mu\right)}+  \sqrt{\epsilon} +\sqrt{\delta_n} \\ 		
			& \leq {2 \gamma^{k}\left(\mu, \mu_n\right)+\sqrt{\epsilon}+\sqrt{\delta_n}},
		\end{flalign*}

		\noindent where we obtain the last inequality of the previously lemma.
		
	\end{proof}

%
%

\begin{lemma}
	\label{lemma-3}
	Under the previously conditions, we can establish the following converge results with $L_2$-Wasserstein distance:
	\begin{enumerate}
		\item $\lim_{n\to \infty} \gamma^{k}\left(\mu, \mu_n\right)=0$ with probability  one and
		\item $\lim_{n\to \infty} \mathbb{E}\left(\gamma^{K}\left(\mu, \mu_n\right)\right)=0$.
	\end{enumerate}
\end{lemma}
		
\begin{proof}
	A Weak Law of Large Numbers for Empirical Measures \cite{bogachev2018weak} guarantee the convergence $\mu_n \to \mu$. In addition, we know by Skorokhod's representation theorem  that exist $Y_n \stackrel{}{\sim} \mu_n$ and   $Y \stackrel{}{\sim} \mu$ of such way that $Y_n\to Y$ with probability equal one. By virtue of continuous mapping theorem mapping is verified also that $\phi\left(Y_n\right)\to \phi\left(Y\right)$. Then, with the application of triangle inequality, we can establish that:
	\begin{multline*}
		2\norm{\phi\left(Y_n\right)}_{k}^{2}+2\norm{\phi\left(Y\right)}_{k}^{2} - \norm{\phi\left(Y_n\right)-\phi\left(Y\right)}_{k}^{2} \\
		\geq  \norm{\phi\left(Y_n\right)}_{k}^{2}+\norm{\phi\left(Y\right)}_{k}^{2}- \norm{\phi \left(Y_n\right)}_{k} \norm{\phi\left(Y\right)}_{k} 
		\geq 0
	\end{multline*}
	
	Then, Fatou's lemma implies that:
	
	\begin{align*}
		\lim_{n\to \infty} \inf \mathbb{E}\left(2\norm{\phi\left(Y_n\right)}^{2}_{k} + 2\norm{\phi\left(Y\right)}^{2}_{k} - \norm{\phi\left(Y\right)-\phi\left(Y_n\right)}^{2}_{k}\right) \geq 4\mathbb{E}(\norm{\phi\left(Y\right)}^{2}_{k}).
	\end{align*}

	Since $\mathbb{E}\left(\phi\left(Y_n\right)^{2}\right)= \lim_{n\to \infty} \frac{1}{n} \sum_{i=1}^{n} \phi\left(X_i\right)^{2}= \mathbb{E}\left(\phi\left(X\right)\right)^{2}= \mathbb{E}\left(\phi\left(Y\right)\right)^{2}$, we deduce that $ \lim_{n\to \infty} \mathbb{E}(\norm{\phi\left(Y_n\right)-\phi\left(Y\right)}^{2}_{\mathbb{K}})=0$, which implies that $\lim_{n\to \infty} \gamma^{k}\left(\mu, \mu_n\right)=0$ with probability equal one.
			
	\noindent For the second part of the proof, let $\Pi\left(\mu,\mu_n\right)$ be the set of joint probability measures on $\mathbb{H}^{p}\times  \mathbb{H}^{p}$ with marginal probabilities $\mu$ and $\mu_n$ respectively. By definition, the $L_2$-Wasserman distance can be written as:
	\begin{align}
		\gamma^{K}\left(\mu, \mu_n \right) = \inf_{\substack{\pi \in \Pi{\left(\mu,\mu_{n}\right)}}} \int \left( \norm{\phi\left(x\right) 
		-\phi\left(y\right)}_{k}^{2} d\pi\left(x,y\right)\right)^{1/2}. 
	\end{align}
	%
	Let $C>0$ be an arbitrary non-negative constant, and let $\mathcal{D}$ be the subset of $\mathbb{H}^{p}\times \mathbb{H}^{p}$ defined by:
	\begin{align}
		\mathcal{D} = \left\{ \left(x,y\right)\in \mathbb{H}^{p}\times \mathbb{H}^{p}:	 \max(\norm{ \phi\left(x\right)}_{k},\norm{\phi\left(y\right)}_{k}) \leq C \right\} 
	\end{align}
	%
	%
	%
	For any $\pi \in \Pi\left(\mu,\mu_n\right)$, we have:
	
	\begin{flalign*}
		\int \norm{\phi\left(x\right) - \phi\left(y\right)}^{2}_{k} d\pi\left(x,y\right)  &= \int_{\mathcal{D}} \norm{\phi\left(x\right)-\phi\left(y\right)}^{2}_{k} d\pi\left(x,y\right) 
		+ \int_{\mathcal{D}^{C}} \norm{\phi\left(x\right)-\phi\left(y\right)}^{2}_{k} d\pi\left(x,y\right)  \\
	    &\leq 
		\int_{\mathcal{D}} \norm{\phi\left(x\right)-\phi\left(y\right)}^{2}_{k} d\pi\left(x,y\right)\\ 
		& \quad + 2\int_{\mathcal{D}^{c}} \norm{\phi\left(x\right)}^{2}_{k}  1\{\norm{\phi\left(x\right)}_{k}>C\} \, d\mu\left(x\right) \\
		& \quad + 2 \int_{\mathcal{D}^{c}} \norm{\phi\left(y\right)}^{2}_{k}  1\{\norm{\phi\left(y\right)}_{k}>C\} d\mu_n\left(y\right)\\
		& \quad + 2\int_{\mathcal{D}^{c}} \norm{\phi\left(x\right)}^{2}_{k}  1\{\norm{\phi\left(x\right)}_{k}\leq C, \norm{\phi\left(y\right)}_{k}>C \} d\pi\left(x,y\right)\\
		& \quad + 2\int_{\mathcal{D}^{c}} \norm{\phi\left(y\right)}_{k}^{2}  1\{\norm{\phi\left(y\right)}_{k}\leq C, \norm{\phi\left(x\right)}_{k} > C \} d\pi\left(x, y\right). 
	\end{flalign*}
			
	\noindent With the application of Markov's inequality and taking the infimun on $\Pi\left(\mu,\mu_n\right)$ two sides of the previously inequality, we can see that:
	\begin{flalign*}
		\mathbb{E}\left( \gamma^{K}\left(\mu,\mu_n\right)\right) & \leq 
		\mathbb{E} \left[ \inf_{\pi \in \Pi{\left(\mu,\mu_n\right)}} \int_{\mathcal{D}} \left( \norm{\phi\left(x\right)-\phi\left(y\right)}^{2}_{k} d\pi \left(x,y\right) \right)^{1/2} \right. \\
		& \left. \quad + 8\int_{\mathcal{D}^{c}} \norm{\phi \left( x \right)}^{2}_{k} 1\{\norm{\phi\left( x \right)}_{k}>C\}  d \mu\left(x\right)\right]. 
	\end{flalign*}
	
	For a fixed $C\geq  0$, the first term of right-hand side goes to $0$ as $n\to \infty$ according to the first part of this lemma and the Lebesgue dominated theorem. Since in our initial assumptions, $\mathbb{E}\left(\norm{X}^{2}\right)\leq \infty$, the second term go to $0$ so $C\to \infty$.   
	Now, in virtue of the bound establish in Lemma \ref{lemma-2} and the last results we can deduce as consequence this theorem.   
\end{proof}

\section{Analysis sensibility of kernel bandwidth parameters}

To analyze the impact of kernel bandwidth parameters in our $AKKB$ algorithm, we use a grid of potential candidates of $\sigma_{data}$, $\sigma_{variables}$ using as a reference the median heuristic. Then, we evaluate the model performance. For this purpose, we use   some genetic expression datasets of the \textit{Schliep lab} repository\footnote{Available at  \url{https://schlieplab.org/Static/Supplements/CompCancer/datasets.htm}} on which true structures are annotated. We consider the following datasets in the analysis: Alizadeh2, Bibtner2,  Chen2, Golub2, Gordon2, Shipp2, Singh2, and West2. 

Bellow, we show different graphics of the results in the Figures \ref{fig:grafico1}-\ref{fig:grafico8}. The axis-$z$ of graphic, represent the precision, axis-$x$,  $\sigma_{data}$, axis-$y$ $\sigma_{variable}$. We can see that the selected parameter of the kernel are critical in the accuracy of our AKKB Algorithm at the data and covariate level and can modify the precision results in more than 40 $\%$ percent in some cases concerning suboptimal parameter configuration.

\begin{figure}[ht!]
	\centering
	\includegraphics[width=0.7\linewidth]{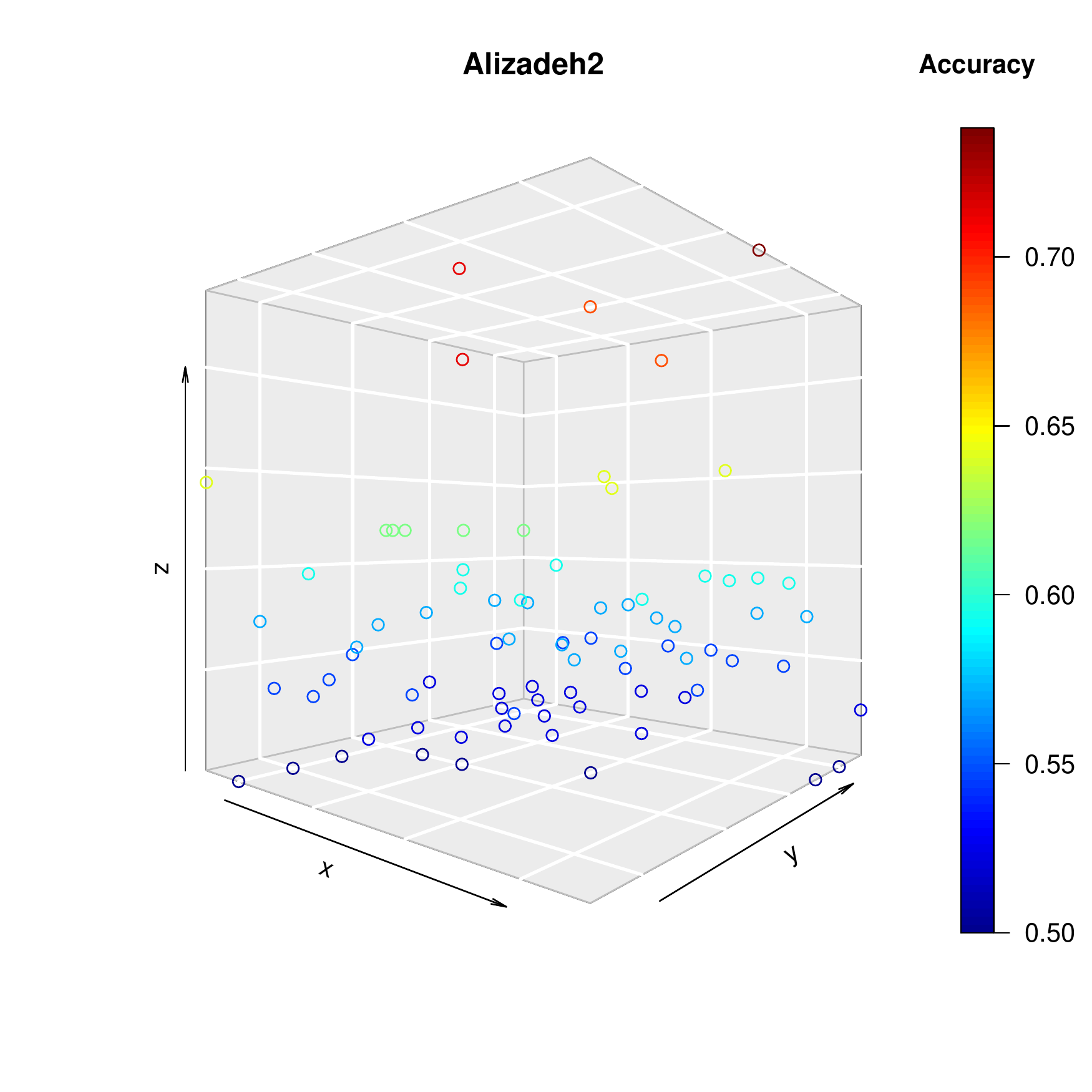}
	\caption{Results of applying our Biclustering algorithm in the dataset Alizadeh2 varying $\sigma_{data}$ and $\sigma_{variables}$ in a grid of potential values according to median heuristic criteria.}
	\label{fig:grafico1}
\end{figure}

\begin{figure}[ht!]
	\centering
	\includegraphics[width=0.7\linewidth]{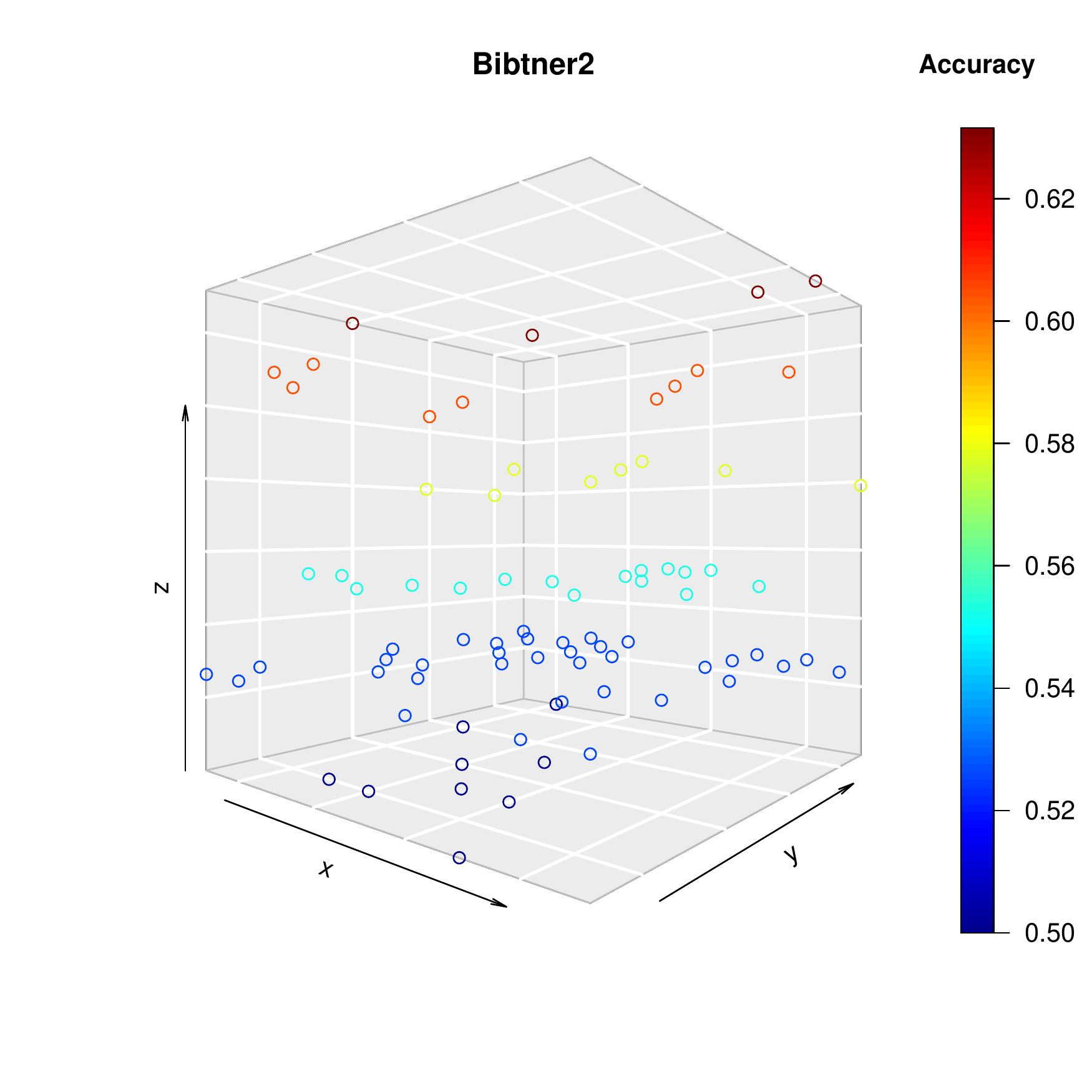}
	\caption{Results of applying our Biclustering algorithm in the dataset Bibtner2 varying $\sigma_{data}$ and $\sigma_{variables}$ in a grid of potential values according to median heuristic criteria.}
	\label{fig:grafico2}
\end{figure}

\begin{figure}[ht!]
	\centering
	\includegraphics[width=0.7\linewidth]{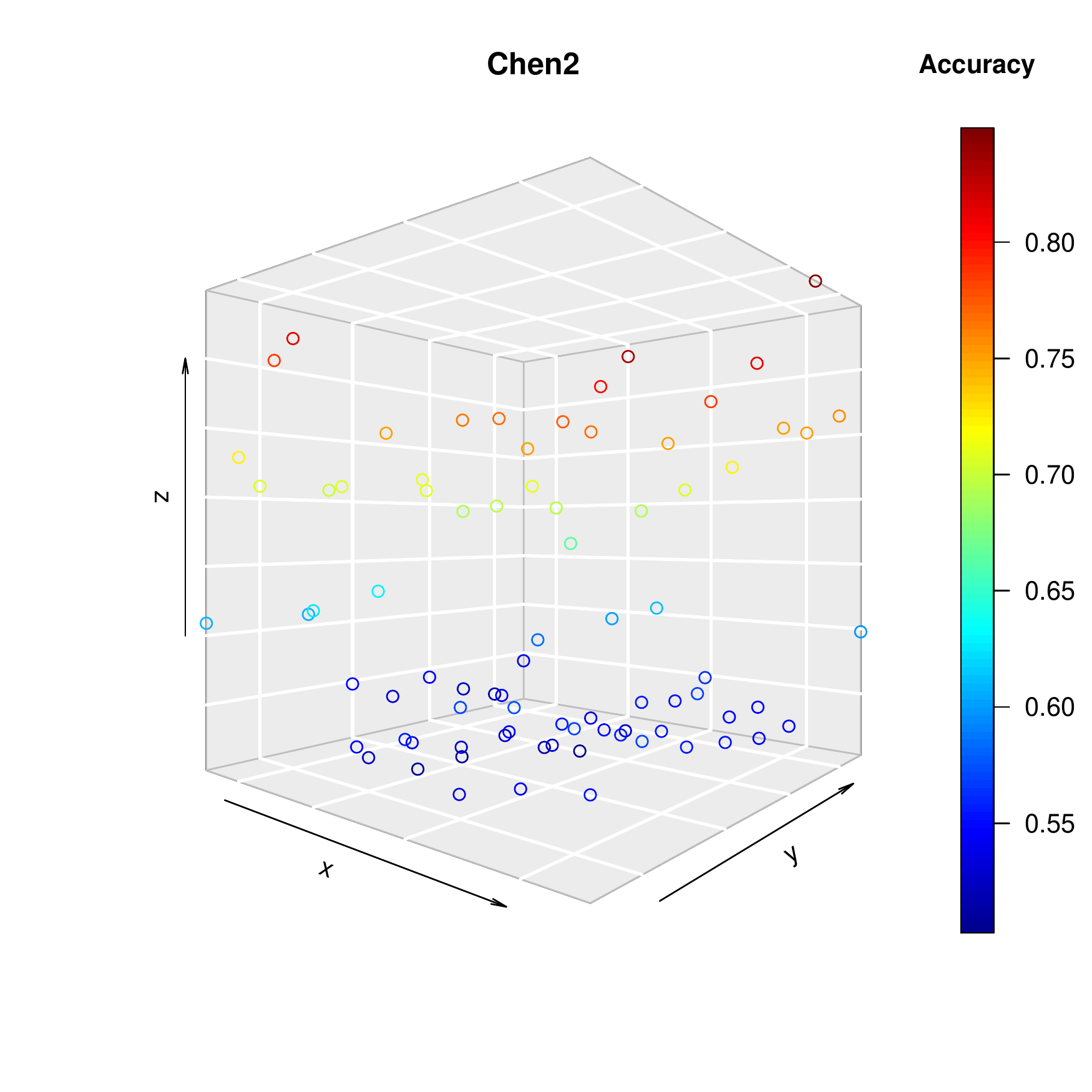}
	\caption{ Results of applying our Biclustering algorithm in the dataset Chen2 varying $\sigma_{data}$ and $\sigma_{variables}$ in a grid of potential values according to median heuristic criteria.}
	\label{fig:grafico3}
\end{figure}

\begin{figure}[ht!]
	\centering
	\includegraphics[width=0.7\linewidth]{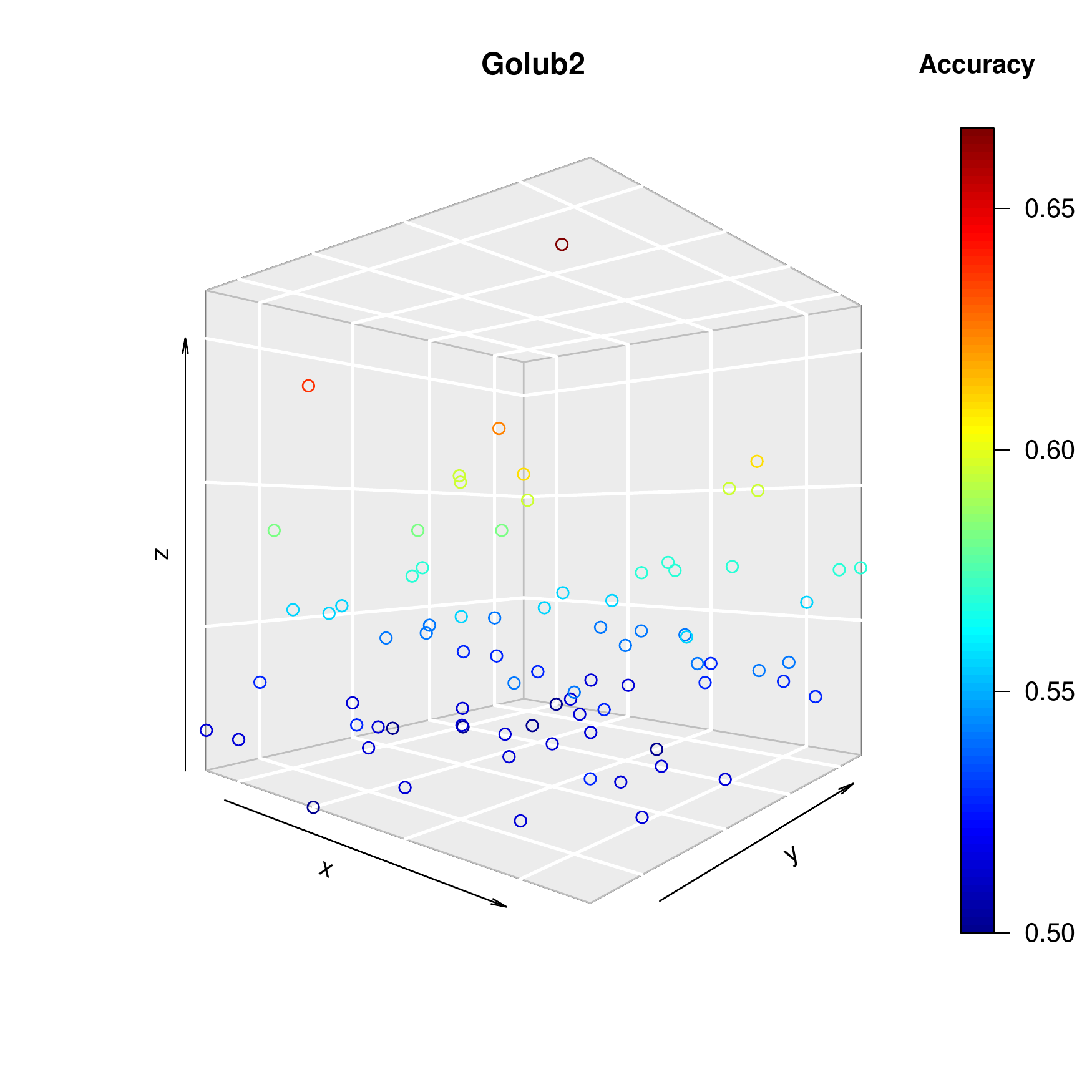}
	\caption{Results of applying our Biclustering algorithm in the dataset Golub2 varying $\sigma_{data}$ and $\sigma_{variables}$ in a grid of potential values according to median heuristic criteria.}
	\label{fig:grafico4}
\end{figure}

\begin{figure}[ht!]
	\centering
	\includegraphics[width=0.7\linewidth]{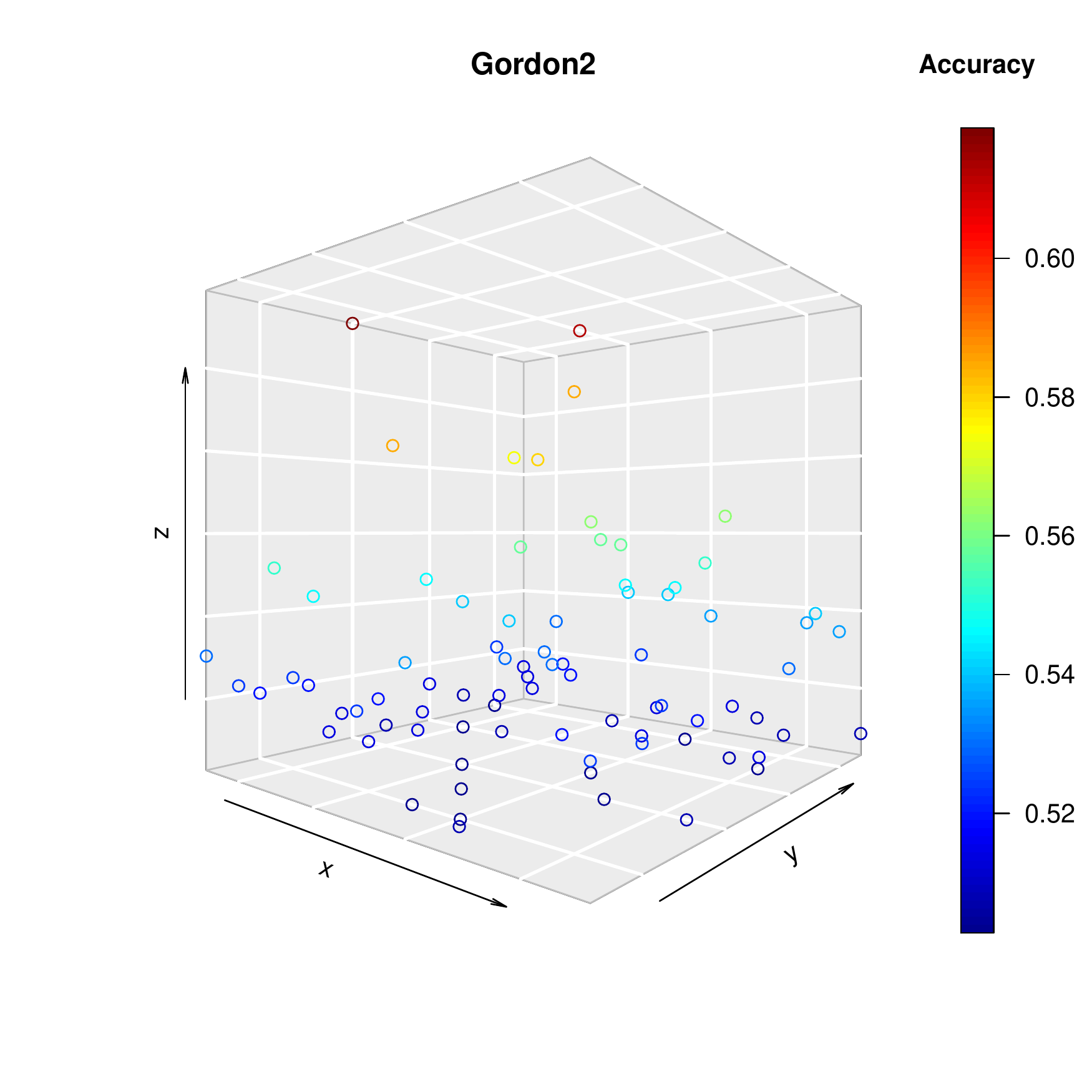}
	\caption{Results of applying our Biclustering algorithm in the dataset Gordon2 varying $\sigma_{data}$ and $\sigma_{variables}$ in a grid of potential values according to median heuristic criteria.}
	\label{fig:grafico5}
\end{figure}

\begin{figure}[ht!]
	\centering
	\includegraphics[width=0.7\linewidth]{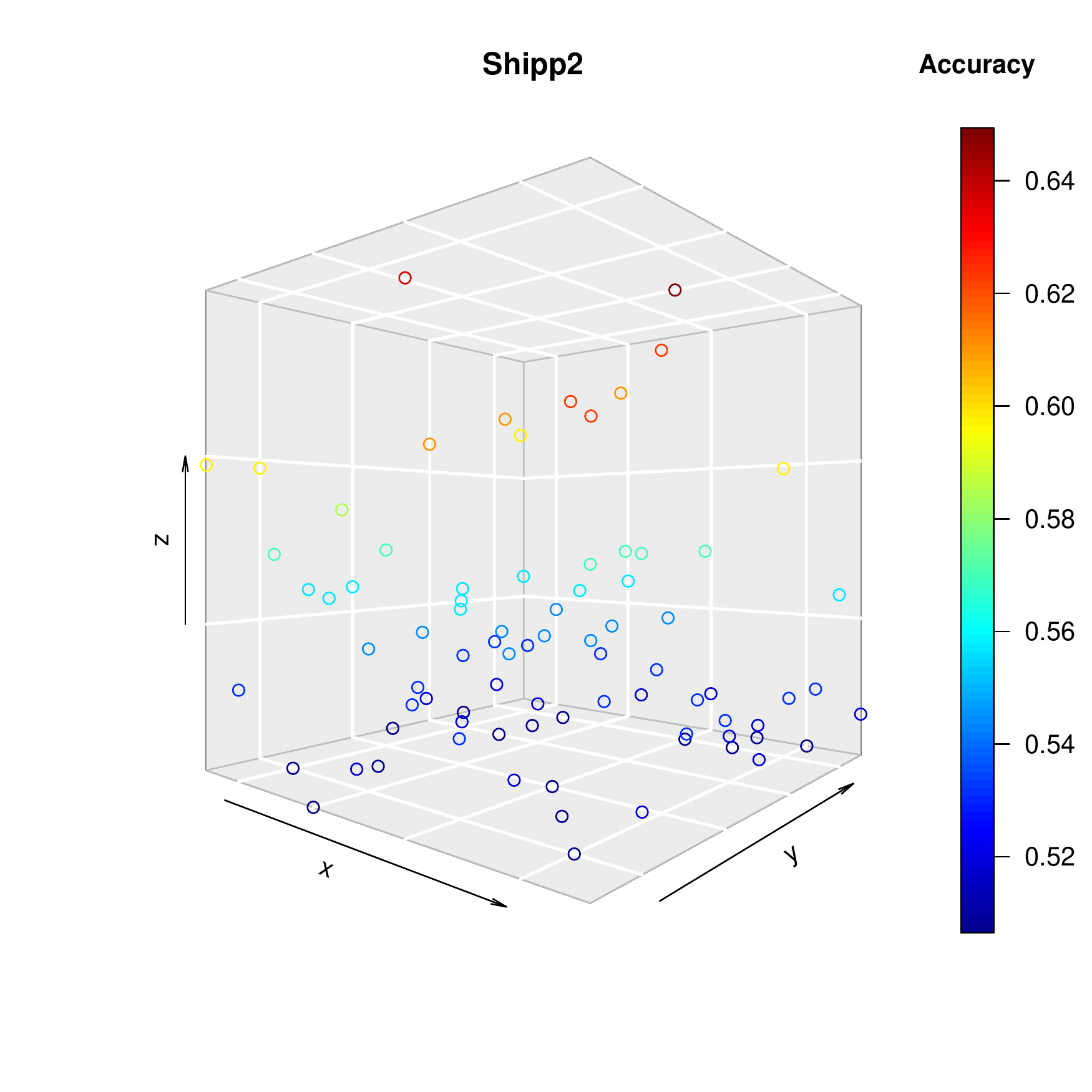}
	\caption{Results of applying our Biclustering algorithm in the dataset Shipp2 varying $\sigma_{data}$ and $\sigma_{variables}$ in a grid of potential values according to median heuristic criteria.}
	\label{fig:grafico6}
\end{figure}

\begin{figure}[ht!]
	\centering
	\includegraphics[width=0.7\linewidth]{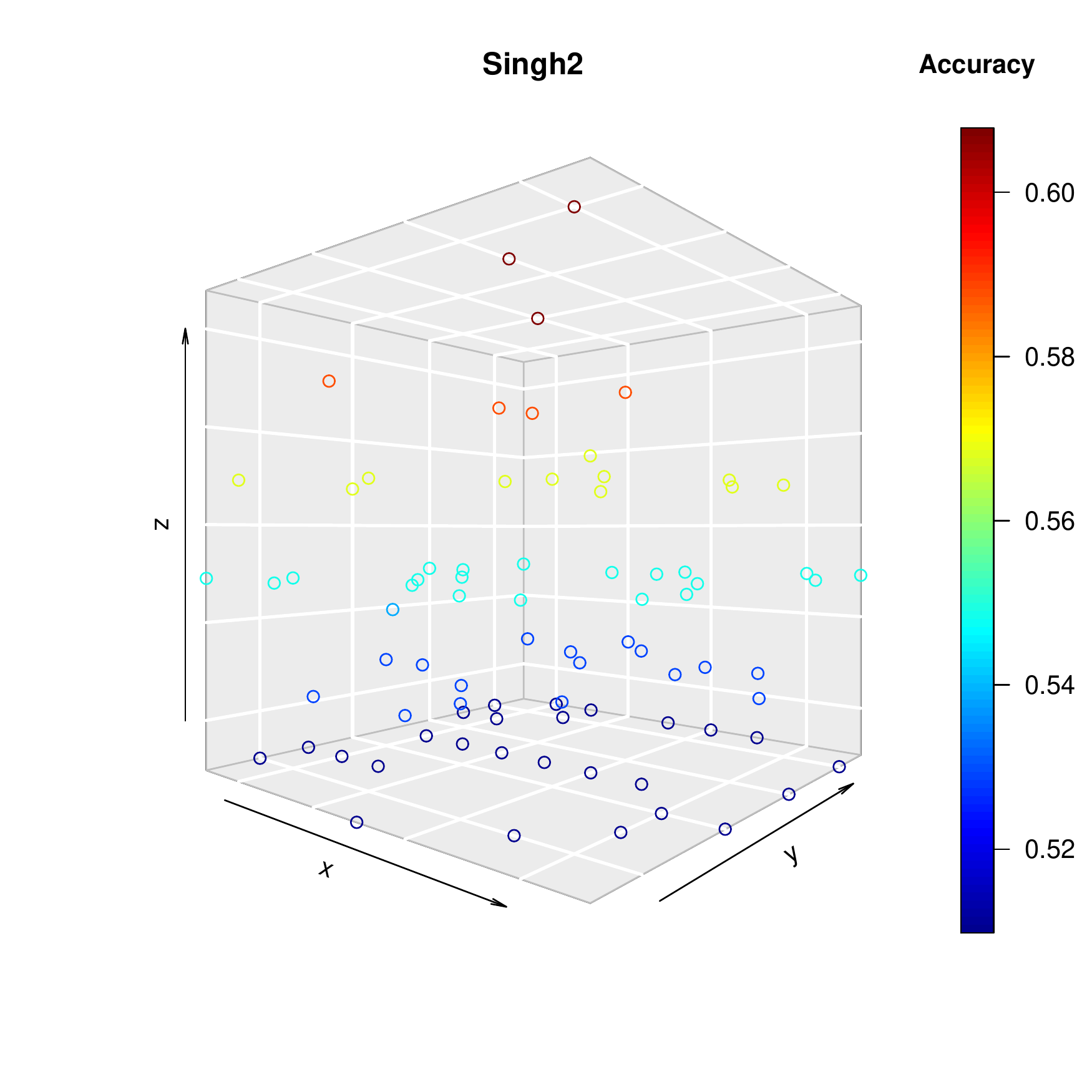}
	\caption{Results of applying our Biclustering algorithm in the dataset Singh2 varying $\sigma_{data}$ and $\sigma_{variables}$ in a grid of potential values according to median heuristic criteria.}
	\label{fig:grafico7}
\end{figure}

\begin{figure}[ht!]
	\centering
	\includegraphics[width=0.7\linewidth]{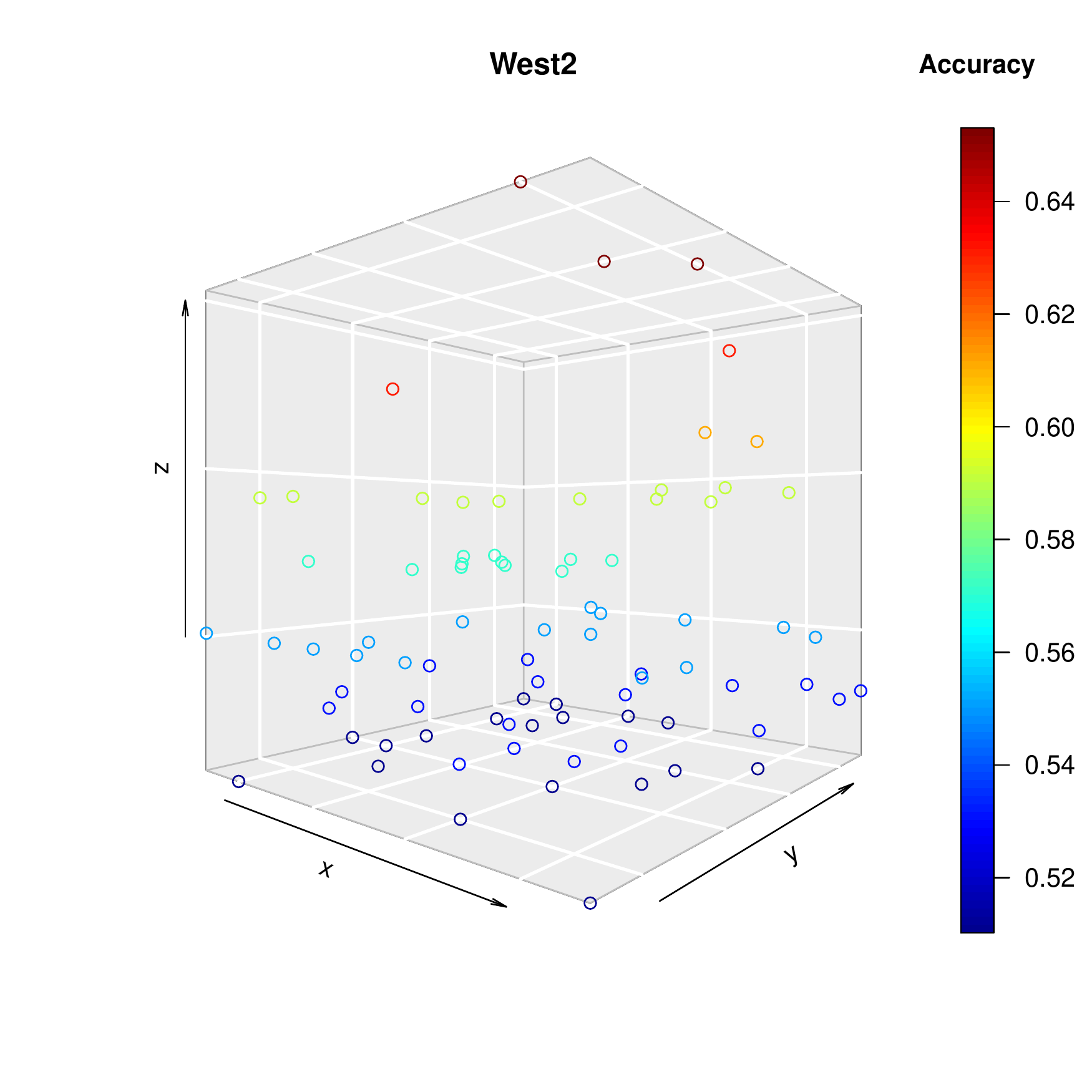}
	\caption{Results of applying our Biclustering algorithm in the dataset West2 varying $\sigma_{data}$ and $\sigma_{variables}$ in a grid of potential values according to median heuristic criteria.}
	\label{fig:grafico8}
\end{figure}


%
%
\end{appendix}

\bibliographystyle{plain}      
\bibliography{arxiv.bib}

	
\end{document}